\documentclass[11pt, a4paper, twoside, reqno]{amsart}

\usepackage[english]{babel}

\usepackage{geometry}

\usepackage[noadjust]{cite}

\usepackage{amssymb}
\usepackage{amsfonts}
\usepackage{amsmath}
\usepackage{latexsym}

\usepackage[ruled, linesnumbered]{algorithm2e}

\usepackage{scalerel}

\usepackage{url}
\usepackage{hyperref}
\usepackage{cleveref}

\usepackage{amsthm}

\usepackage{tabularx}

\usepackage{enumitem}

\makeatletter
\newcommand*\bigcdot{\mathpalette\bigcdot@{.75}}
\newcommand*\bigcdot@[2]{\mathbin{\vcenter{\hbox{\scalebox{#2}{$\m@th#1\bullet$}}}}}
\makeatother

\hypersetup{
    colorlinks = true,
    citecolor = black,
    filecolor = black,
    linkcolor = black,
    urlcolor = black
}

\newtheoremstyle{TheoremNum}
		{\topsep}{\topsep}					
		{\itshape}							
		{}								
		{\bfseries}							
		{.}								
		{ }								
		{\thmname{#1}\thmnote{ \bfseries #3}}	

\theoremstyle{TheoremNum}

\theoremstyle{plain}
\newtheorem{lemma}{Lemma}[section]

\newtheorem{corollary}[lemma]{Corollary}

\theoremstyle{remark}

\newtheorem{remark}[lemma]{Remark}

\theoremstyle{definition}
\newtheorem{definition}[lemma]{Definition}
\newtheorem{example}[lemma]{Example}

\newcommand{\intInterval}[2]{[\![#1,#2]\!]}

\def\BoolSet{\textsc{Bool}}
\def\True{\textsc{True}}
\def\False{\textsc{False}}

\def\LCVP{LCVP}

\def\genDomSym{\lambda}

\def\ColorSet{\textsc{Colors}}

\def\NumberOfColors{q}

\def\WeightSet{\textsc{Weights}}

\def\wlesseq{\preceq}

\def\minWeights{\min}

\def\Error{\textsc{Error}}

\def\weightsSum{\circledast}
\def\bigweightsSum{\bigcircledast}

\def\NumberOfLabels{k}

\newcommand{\notFinalLabels}[1]{\overline{\ell(#1)}}

\def\upcomingNeighbors{n}

\def\tcheck{t_{check}(|V(G)|, q, \mathcal{N})}

\def\SizeOfColorClass{s}

\newcommand{\defProb}[5]{
	\noindent{\renewcommand{\arraystretch}{1}\setlength{\tabcolsep}{2pt}
		\begin{tabularx}{\textwidth}{lX}
			\multicolumn{2}{X}{\sc #1}			\\ 
			\emph{{#2}:}	& \multicolumn{1}{X}{#3}	\\ 
			\emph{{#4}:}	& \multicolumn{1}{X}{#5}	\\ 
		\end{tabularx}
	}
	\smallskip
}

\newcommand{\defProbEnglish}[3]{
	\defProb{#1}{Instance}{#2}{Question}{#3}
}

\newcommand{\setst}[2]{\{ #1 : #2 \}}

\DeclareMathOperator*{\bigcircledast}{\scalerel*{\circledast}{\sum}}

\newcommand{\quantifierFormula}[3]{#1 #2 ,\, #3}

\newcommand{\forallFormula}[2]{\quantifierFormula{\forall}{#1}{#2}}

\begin{document}
\title{Locally checkable problems parameterized by clique-width}

\author
	[N. Baghirova]
	{Narmina Baghirova}
\address{
	University of Fribourg, Department of Informatics, Fribourg, Switzerland
}
\email{
	narmina.baghirova@unifr.ch
}

\author
	[C.L. Gonzalez]
	{Carolina Luc\'ia Gonzalez}
\address{
	CONICET-Universidad de Buenos Aires.
	Instituto de Investigaci\'on en Ciencias de la Computaci\'on (ICC).
	Buenos Aires, Argentina.
}
\email{
	cgonzalez@dc.uba.ar
}

\author
	[B. Ries]
	{Bernard Ries}
\address{
	University of Fribourg, Department of Informatics, Fribourg, Switzerland
}
\email{
	bernard.ries@unifr.ch
}

\author
	[D. Schindl]
	{David Schindl}
\address{
	University of Fribourg, Department of Informatics, Fribourg, Switzerland
}
\email{
	david.schindl@unifr.ch
}

\begin{abstract}
We continue the study initiated by Bonomo-Braberman and Gonzalez in 2020 on $r$-locally checkable problems.
We propose a dynamic programming algorithm that takes as input a graph with an associated clique-width expression and solves a $1$-locally checkable problem under certain restrictions.
We show that it runs in polynomial time in graphs of bounded clique-width, when the number of colors of the locally checkable problem is fixed.
Furthermore, we present a first extension of our framework to global properties by taking into account the sizes of the color classes, and consequently enlarge the set of problems solvable in polynomial time with our approach in graphs of bounded clique-width.
As examples, we apply this setting to show that, when parameterized by clique-width, the $[k]-$Roman domination problem is FPT, and the $k$-community problem, Max PDS and other variants are XP.
\end{abstract}



\maketitle

\section{Introduction}
\label{sec:introduction}

Many graph problems can be stated as a sort of partitioning, or equivalently, as a sort of coloring problem. Furthermore, most decision problems on graphs from the literature belong to the class NP, and their certificate verification algorithms often consist in checking some \textit{local} property for each vertex, i.e. involving itself and its neighborhood only, plus possibly some \textit{global} property concerning, for instance, the sizes or the connectivity of some subsets of vertices.
One could therefore try to cover a broad variety of these problems under a same umbrella, and hence develop efficient algorithms to solve them at once. With this objective in mind, several definitions of subsets of partitioning problems, where each vertex has to satisfy a local property, as well as extensions of these sets of problems including some global property, have been proposed and shown to be solvable in polynomial time in various graph classes.
In particular, in~\cite{LCPTreewidth}, the authors defined so-called \emph{$r$-locally checkable problems}.
Each of these problems has an associated set of colors and a \emph{check function}, that is, a function that takes as input a vertex $v$ of the graph and a coloring of the $r$-neighborhood of $v$ (i.e. the set of vertices at distance at most $r$ from $v$) and outputs $\True$ or $\False$.
A \emph{proper coloring} of the input graph $G$ is defined as a coloring $c$ of the vertices such that, for every vertex $v$, the check function applied to $v$ and the restriction of $c$ to the $r$-neighborhood of $v$ outputs $\True$.
They also consider a set of weights with a total order, and associate a weight to each pair of vertex and possible color.
The weight of a coloring $c$ is then naturally obtained by combining the weights of the pairs $(v,c(v))$.
Then, an $r$-locally checkable problem consists in finding the minimum weight of a proper coloring of the input graph $G$. Examples of $r$-locally checkable problems include {\sc $k$-Coloring}, {\sc Maximum Independent Set} and {\sc Minimum Dominating Set}~\cite{LCPTreewidth}.

Since many $r$-locally checkable problems are hard on general graphs, it is of interest to determine under which conditions (on the check function and the set of colors) we can efficiently solve them for a given class of graphs.
In~\cite{LCPTreewidth}, the authors showed that, under mild conditions, $r$-locally checkable problems can be solved in polynomial time in graphs of bounded tree-width.
In this paper, we will focus on 1-locally checkable problems with an associated \emph{color-counting} check function, defined as follows.
\begin{definition}
Let $G$ be a graph and $\ColorSet = \{a_1, \ldots, a_{\NumberOfColors}\}$ be a set of colors.
A check function $f$ is \emph{color-counting} if it only depends on the vertex $v$, the color it receives and, for each color $a \in \ColorSet$, the number of neighbors of $v$ of color $a$.

In other words, a check function $f$ is \emph{color-counting} if there exists a function $f'$ such that
\begin{center}
$f(v, c) = f'(v, c(v), \upcomingNeighbors_1, \ldots, \upcomingNeighbors_{\NumberOfColors})$
\end{center}
for every vertex $v \in V(G)$ and every coloring $c$ of the neighborhood of $v$ (denoted by $N_G(v)$), where $\upcomingNeighbors_j = |\setst{u \in N_G(v)}{c(u) = a_j}|$ for all $j \in \{1, \ldots, \NumberOfColors\}$.
\end{definition}
Since we are only going to work with color-counting check functions in this paper, we will directly refer to them as $check(v, a, \upcomingNeighbors_1, \ldots, \upcomingNeighbors_{\NumberOfColors})$.

In~\cite{LCPmimwidth}, the authors analyzed the restrictions on 1-locally checkable problems with respect to mim-width.
They define \emph{$d$-stable} check functions, which are a subset of the color-counting check functions, and we will use them to improve our complexity results.
\begin{definition}[\cite{LCPmimwidth}]
Let $d \in \mathbb{N}$.
Let $G$ be a graph, $\ColorSet$ be a set of $\NumberOfColors$ colors and $check$ be a color-counting check function.
We say that $check$ is \emph{$d$-stable} if for all $v \in V(G)$, $a \in \ColorSet$ and non-negative integers $\upcomingNeighbors_1, \ldots, \upcomingNeighbors_{\NumberOfColors}$ we have
\begin{center}
	$check(v, a, \upcomingNeighbors_1, \ldots, \upcomingNeighbors_{\NumberOfColors}) = check(v, a, \min(d, \upcomingNeighbors_1), \ldots, \min(d, \upcomingNeighbors_{\NumberOfColors}))$.
\end{center}
\end{definition}

We present a dynamic programming algorithm, which is XP parameterized by clique-width, for 1-locally checkable problems with a constant number of colors and a color-counting check function.
Moreover, this algorithm is FPT when the check function is also $d$-stable, for any constant $d$.
In a second step, we extend our framework in such a way that we are able to ensure that the size of a given color class belongs to some predefined set of integers.
By including this global property for as many colors classes as necessary, the application of our framework allows to obtain first XP algorithms parameterized by clique-width for problems such as \textsc{$k$-community}, \textsc{Max PDS} and \textsc{(global) [$k$]-Roman Domination}, as well as some variants of them.
A generalization of this framework to $r$-locally checkable problems, for any fixed $r$, would be quite natural, and the authors of this paper are currently working on it.

The set of locally checkable problems considered here is a subset of the one considered in~\cite{LCPTreewidth}, but notice that our assumptions above are not too restrictive.
Indeed, if we do not impose these assumptions then, as explained in~\cite{LCPTreewidth}, one obtains locally checkable problems that are NP-hard on complete graphs (which have clique-width 2) and thus, it is unlikely to find XP algorithms parameterized by clique-width for this more general class of locally checkable problems.

As mentioned above, several definitions of subsets of partitioning problems have been defined in the literature and shown to be solvable in polynomial time in various graph classes.
We cite here some of the corresponding publications that are the most closely related to our work.

In~\cite{LCVSVP-paper2,LCVP-paper3,LCVSVP-paper}, the authors studied a large class of vertex partitioning problems called \emph{locally checkable vertex partitioning} ({\LCVP}) problems.
In these problems, a $q \times q$ matrix $D$ is given, where each entry is a finite or cofinite set of integers.
A partition of the set of vertices $V_1,\ldots, V_q$ is sought, such that for each $i,j\in \{1, \ldots, q\}$, we have $|N(v)\cap V_j|\in D[i,j]$ for all $v\in V_i$.
Empty partition classes are allowed.
In~\cite{LCVSVP-paper}, Telle and Proskurowski solved these problems in polynomial time on graphs of bounded treewidth.
This result was generalized in~\cite{LCVSVP-paper2}, where Bui-Xuan, Telle and Vatshelle gave an algorithm that solves LCVP problems given a decomposition tree of the input graph.
In the same paper, they proved that this algorithm is FPT parameterized by boolean-width, and later in~\cite{LCVP-paper3}, Jaffke et al. showed that the same algorithm is XP parameterized by mim-width, when a suitable decomposition tree is given.
As shown in~\cite{LCPmimwidth}, every LCVP problem can be modeled as a 1-locally checkable problem with a $d$-stable check function (where $d$ is as defined in~\cite{LCVSVP-paper2}):
\begin{center}
	$check(v, a, \upcomingNeighbors_1, \ldots, \upcomingNeighbors_q) = \left(\forallFormula{j \in \{1, \ldots, q\}}{\upcomingNeighbors_j \in D[a, j]}\right).$
\end{center}
While many locally checkable problems are expressible as LCVP problems, there are still some relevant problems that do not admit such a characterization, but do belong to the set of problems we analyze in this paper.
Examples include \textsc{$[k]-$Roman domination} and \textsc{balanced $k$-community}, see Section~\ref{sec:Applications}.

In~\cite{GERBER2003719}, Gerber and Kobler studied a variation of {\LCVP}, with two modifications. On one hand they restrict the entries of $D$ to sets of consecutive integers, and on the other hand, they associate to each vertex $v$ a set $\rho(v)\subseteq \{1, \ldots, q\}$ such that $v\in V_i\Rightarrow i\in \rho(v)$.
They show that the problems in this framework are XP when parameterized by clique-width. Notice that these problems are also covered by our framework.

In~\cite{courcelle_boundedCliquewidth}, Courcelle, Makowsky and Rotics presented an algorithm which, given as input a graph with an associated clique-width expression, solves problems expressible in a certain variation of Monadic Second-Order Logic, called MSO$_1$.
On graphs of clique-width at most $k$, the running time of their algorithm is linear in the size of the input graph.
However, as pointed out in~\cite{FRICK20043}, the multiplicative constant grows extremely fast with $k$.

Following a similar research line, in their recent article~\cite{DN-logic}, Bergougnoux, Dreier and Jaffke defined an extension of existential MSO$_1$, which they call \emph{distance neighborhood logic with acyclicity and connectivity constraints} (A\&C DN logic).
They provided an algorithm that solves problems expressible in this logic, given a suitable branch decomposition of the input graph.
The complexity of the algorithm is expressed in terms of the \emph{$d$-neighborhood} equivalence relation (see~\cite{LCVSVP-paper2}), allowing them to state their main result parameterized by mim-width (XP), tree-width, rank-width or clique-width (FPT), with a single-exponential dependence.
As shown in~\cite{LCPmimwidth}, all locally checkable problems with constant number of colors and $d$-stable check functions, for some constant $d$, can be expressed in A\&C DN logic.
However, locally checkable problems with a color-counting check function that is not $d$-stable for any constant $d$, and possibly extended with global properties, such as \textsc{balanced $k$-community}, cannot be directly expressed by an A\&C DN logic formula of fixed length.

This paper is structured as follows.
In Section~\ref{sec:Preliminaries}, we give some definitions and notations.
In Section~\ref{sec:LCP}, we formally present our framework, while in Section~\ref{sec:Algorithm}, we describe the dynamic programming algorithm, prove its correctness and analyse its complexity.
Section~\ref{sec:SizeGlobProp} deals with the extension of our results of the previous section to include the global size property.
Finally, in Section~\ref{sec:Applications}, we apply our results to some selected problems.
Due to space constraints, we omit the proofs and present them in the appendix.

\section{Preliminaries}
\label{sec:Preliminaries}

\subsection{Algebraic definitions}
Let $f \colon D \to C$ be a function and let $S \subseteq D$.
We denote by $f|_S$ the function $f$ restricted to the domain $S$, that is, the function $f|_S \colon S \to C$ is defined as $f|_S(x) = f(x)$ for all $x \in S$. 
Let $D'$ be a set such that $D \cap D' = \emptyset$, and let $g \colon D' \to C$.
We denote by $f \cup g$ the function $h \colon D \cup D' \to C$ such that $h(x) = f(x)$ if $x \in D$, and $h(x) = g(x)$ if $x \in D'$.
Note that, since $D$ and $D'$ are disjoint, $f \cup g$ is well defined.

We denote by $\intInterval{a}{b}$, with $a,b \in \mathbb{Z}$ and $a \leq b$, the set of all integer numbers greater than or equal to $a$ and less than or equal to $b$, that is $\{a, a+1, \ldots, b\}$.
Furthermore, we use $\BoolSet$ to denote the set $\{\True,\False\}$.

\subsection{Graph theoretical definitions}
Throughout this paper, we consider simple, finite and undirected graphs.
For graph theoretical notions not defined here, the reader is referred to~\cite{DouglasWest}.

The notion of clique-width of a graph $G$, denoted by $cw(G)$, was first introduced in~\cite{courcelle_CW_Introduction}. It is defined as the minimum number of labels needed to construct $G$ using the following 4 operations:
\begin{itemize}
\item creation of a new vertex $v$ with label $i$ (denoted by $i(v)$);
\item disjoint union of two labeled graphs $G_1$ and $G_2$ (denoted by $G_1 \oplus G_2$);
\item join between two labels $i$ and $j$, $i \neq j$, i.e. adding an edge between every vertex with label $i$ and every vertex with label $j$ (denoted by $\eta_{i,j}$);
\item renaming of label $i$ to label $j$, i.e. every vertex with label $i$ gets label $j$ (denoted by $\rho_{i \rightarrow j}$).
\end{itemize}

 Given a graph class ${\mathcal G}$, the clique-width of $\mathcal{G}$ is $cw(\mathcal{G}) = \sup\{cw(G) \mid G \in \mathcal{G}\}$.
We say that $\mathcal{G}$ is \emph{of bounded clique-width} if $cw(\mathcal{G}) < \infty$.

A \textit{clique-width expression} is simply a well-formed expression of operations each corresponding to one of the four operations mentioned above.
For a clique-width expression $e$, we denote by $G_e$ the graph constructed by $e$.
If the number of distinct labels used in a clique-width expression $e$ is at most $\NumberOfLabels$, then we say it is a \textit{clique-width $\NumberOfLabels$-expression}.
It was shown in~\cite{irredundantCliquewidth} that any graph $G$ admitting a clique-width $\NumberOfLabels$-expression also admits an \emph{irredundant clique-width $\NumberOfLabels$-expression}, i.e., such that whenever we execute a join operation $\eta_{i,j}$, there are no already existing edges between vertices with label $i$ and vertices with label $j$.

Consider a clique-width expression $e$ and the corresponding graph $G_e$. An \textit{expression tree} of $G_e$ is a rooted binary tree $T_e$ defined as follows:
\begin{itemize}
\item The nodes of $T_e$ are of four types corresponding to operations $i(\cdot)$, $\oplus$, $\eta$ and $\rho$.
\item The leaves of $T_e$ correspond to the creation operation $i(\cdot)$.
\item A disjoint union node $\oplus$ corresponds to the disjoint union of the graphs associated with its two children.
\item A join node $\eta_{i,j}$ corresponds to the graph associated with its unique child in which we make all vertices of label $i$ adjacent to all vertices of label $j$.
\item A relabeling node $\rho_{i\rightarrow j}$ corresponds to the graph associated with its unique child in which we change label $i$ to label $j$.
\item The graph $G_e$ corresponds to the graph associated with the root of $T_e$.
\end{itemize}

Notice that for every node $t \in V(T_e)$, the subtree of $T_e$ rooted at $t$ defines a clique-width expression $e_t$ the corresponding graph of which, denoted by $G_{e_t}$, is a subgraph of $G_e$. We say that $e'$ is a \emph{subexpression} of $e$ if $e'$ is the expression determined by the subtree of $T_e$ rooted at some node $t \in V(T_e)$. Consider any vertex $v$ in $G_{e_t}$ for some $t \in V(T_e)$. Then all neighbors of $v$ in $G_e$ which are not yet neighbors of $v$ in $G_{e_t}$, i.e. the edges between $v$ and these vertices are only defined by the ancestor operations of $t$ in $T_e$, are said to be \emph{upcoming neighbors of $v$ with respect to $e_t$}.
Notice that for any two vertices in $G_{e_t}$ having the same label, their sets of upcoming neighbors with respect to $e_t$ are identical.

Let $e$ be a clique-width $\NumberOfLabels$-expression, $G_e$ be its corresponding graph and let $T_e$ be an expression tree of $G_e$.
We define the function $\ell_e : V(G) \to \intInterval{1}{\NumberOfLabels}$ such that $\ell_e(v)$ is the final label of $v$, i.e. the label of $v$ after the operation corresponding to the root of $T_e$.
We also define $\notFinalLabels{e}$ as the set of labels $i$ such that there exists no $v \in V(G_e)$ such that $\ell_e(v) = i$.

In the remaining of our paper, we will only consider irredundant clique-width $\NumberOfLabels$-expressions where in any relabeling operation $\rho_{i \rightarrow j}(e)$ we have $j\not\in\notFinalLabels{e}$.
Notice that under these assumptions the total number of operations in such a clique-width expression of a graph $G$ is in $O(|V(G)|+|E(G)|)$.

\subsection{Finite-state automata}
A \emph{deterministic finite-state automaton} is a five-tuple $(Q, \Sigma, \delta, q_0, F)$ that consists of
\begin{itemize}
\item $Q$: a finite set of \emph{states},
\item $\Sigma$: a finite set of \emph{input symbols} (often called the \emph{alphabet}),
\item $\delta \colon Q\times \Sigma \rightarrow Q$: a \emph{transition function},
\item $q_0 \in Q$: an \emph{initial} (or \emph{start}) \emph{state}, and
\item $F \subseteq Q$: a set of \emph{final} (or \emph{accepting}) \emph{states}.
\end{itemize}

We say that an automaton $M = (Q, \Sigma, \delta, q_0, F)$ \emph{accepts} a string $c_1 \ldots c_n$, with $n \geq 1$, if and only if $c_i \in \Sigma$ for all $1 \leq i \leq n$ and $\delta(\ldots \delta(\delta(q_0, c_1), c_2) \ldots , c_n) \in F$.

For more about automata theory, we refer the reader to~\cite{HopcroftUllman}.

\subsection{Weight sets}
Let $(\WeightSet, \wlesseq)$ be a totally ordered set with a maximum element (called $\Error$), together with the minimum operation of the order $\wlesseq$ (called $\minWeights$) and a closed binary operation on $\WeightSet$ (called $\weightsSum$) that is commutative and associative, has a neutral element and an absorbing element that is equal to $\Error$, and the following property is satisfied: $s_1 \wlesseq s_2 \Rightarrow s_1 \weightsSum s_3 \wlesseq s_2 \weightsSum s_3$ for all $s_1, s_2, s_3 \in \WeightSet$.
In such a case, we say that $(\WeightSet, \wlesseq, \weightsSum)$ is a \emph{weight set}.

A classic example of a weight set is $(\mathbb{N} \cup \{+\infty\}, \leq, +)$.
Notice that the maximum element is $+\infty$ in this case.
We could also consider the reversed order of natural weights: $(\mathbb{N} \cup \{-\infty\}, \geq, +)$, where the maximum element is now $-\infty$.
Another simple example worth mentioning is $(\{0,1\}, \leq, \max)$.

\section{Color-counting 1-locally checkable problems}\label{sec:LCP}

Suppose we are given:
\begin{itemize}
\item a simple undirected graph $G$,

\item a set $\ColorSet = \{a_1, \ldots, a_\NumberOfColors\}$,

\item for every $v \in V(G)$, a nonempty set $L_{v} \subseteq \ColorSet$ of possible colors for $v$,

\item a weight set $(\WeightSet, \wlesseq, \weightsSum)$,

\item for every $v \in V(G)$ and for every $a \in L_{v}$, a weight $\textsc{w}_{v, a} \in \WeightSet - \{\Error\}$ of assigning color $a$ to vertex $v$, and

\item a color-counting check function $check$.
\end{itemize}

We say that a coloring $c : V(G) \to \ColorSet$ is \emph{valid} if $c(v) \in L_v$ for all $v \in V(G)$.
The weight of a valid coloring $c$ is $\textsc{w}(c) = \bigweightsSum_{v \in V} \textsc{w}_{v, c(v)}$.
Furthermore, we say that $c$ is a \emph{proper} coloring of $G$ if it is a valid coloring of $G$ and $check(v, c(v), \upcomingNeighbors_1, \ldots, \upcomingNeighbors_{\NumberOfColors})$ is true for every $v \in V(G)$, where $\upcomingNeighbors_j = |\setst{u \in N_G(v)}{c(u) = a_j}|$ for all $j \in \intInterval{1}{\NumberOfColors}$.

A \emph{color-counting 1-locally checkable problem} consists in finding the minimum weight of a proper coloring of the input graph $G$.

\section{Algorithm}\label{sec:Algorithm}

Consider a color-counting 1-locally checkable problem $\Pi$ and let $G$ be the input graph and $e_G$ a clique-width $\NumberOfLabels$-expression of $G$.
Let $\mathcal{N} \in \intInterval{1}{|V(G)|}$ be an integer such that
$check(v, a, \upcomingNeighbors_1, \ldots, \upcomingNeighbors_{\NumberOfColors}) = check(v, a, \min(\mathcal{N}, \upcomingNeighbors_1), \ldots, \min(\mathcal{N}, \upcomingNeighbors_{\NumberOfColors}))$
for all $v \in V(G)$, $a \in \ColorSet$ and non-negative integers $\upcomingNeighbors_1, \ldots, \upcomingNeighbors_{\NumberOfColors}$.

In this section, we present an algorithm which computes the minimum weight of a proper coloring of $G$ by using the expression $e_G$ as well as the notion of $(C,N)$-coloring defined hereafter.

\begin{definition}[$(C, N)$-coloring]\label{def:CN-coloring}
Let $e$ be a subexpression of $e_G$, and let $C$ and $N$ be two matrices in $\intInterval{0}{\mathcal{N}}^{\NumberOfLabels \times \NumberOfColors}$.
A valid coloring $c$ of $G_e$ is called a \emph{$(C, N)$-coloring of $G_e$} if the following two conditions hold:
\begin{enumerate}[label=(C\arabic*)]
\item $\min(\mathcal{N}, |\setst{v \in V(G_e)}{c(v) = a \land \ell_e(v) = i}|) = C[i,a]$ for all $i \in \intInterval{1}{\NumberOfLabels}$ and all $a \in \ColorSet$;
\item for all $v$ in $G_e$ we have $check(v, c(v), \upcomingNeighbors_1, \ldots, \upcomingNeighbors_{\NumberOfColors}) = \True$, where $\upcomingNeighbors_j = \min(\mathcal{N}, N[\ell_e(v), a_j] + |\{u \in N_{G_e}(v) : c(u) = a_j\}|)$ for every $j \in \intInterval{1}{\NumberOfColors}$.
\end{enumerate}
\end{definition}
The minimum weight among all possible $(C,N)$-colorings of $G_e$ is denoted by $\lambda(e,C,N)$, i.e. $\lambda(e,C,N) = \min\setst{\textsc{w}(c)}{c \text{ is a } (C,N)\text{-coloring of } G_e}$.
Notice that if no such coloring exists then $\lambda(e, C, N) = \Error$.

The following lemma explains the link between proper colorings and $(C,N)$-colorings.

\begin{lemma}\label{lem:initialLemma}
Let $\Pi$ be a color-counting 1-locally checkable problem with input graph $G$ and let $e_G$ be a clique-width $\NumberOfLabels$-expression of $G$.
Then the minimum weight of a proper coloring of $G$ equals the minimum among all $\lambda(e_G, C, N_0)$, where $N_0 \in \intInterval{0}{\mathcal{N}}^{\NumberOfLabels \times \NumberOfColors}$ is the matrix whose elements are all $0$ and $C \in \intInterval{0}{\mathcal{N}}^{\NumberOfLabels \times \NumberOfColors}$ is any matrix such that $C[i, a] = 0$ for every $i\in \overline{\ell(e_G)}$ and every $a \in \ColorSet$.

\end{lemma}


So Lemma~\ref{lem:initialLemma} tells us that in order to solve a color-counting 1-locally checkable problem $\Pi$, i.e. in order to find a minimum weight of a proper coloring, it is sufficient to find the minimum weight among all $(C,N_0)$-colorings of the input graph $G$, where $C$ and $N_0$ are as described above.
Our algorithm is based exactly on this idea, i.e. it determines the minimum among all $\lambda(e_G, C, N_0)$.
This is achieved by traversing the binary rooted tree $T_{e_G}$ in a bottom-up fashion and determining in a recursive way the values $\lambda(e, C, N)$, where $e$ is a subexpression of $e_G$ and $C,N\in \intInterval{0}{\mathcal{N}}^{\NumberOfLabels \times \NumberOfColors}$.
Throughout this recursion, the matrices $C$ and $N$ will intuitively behave in the following way: if we have a proper coloring $c$ of $G$ such that $c|_{V(G_e)}$ is a $(C,N)$-coloring of $G_e$, then
\begin{itemize}
	\item $C[i,a]$ represents the minimum between $\mathcal{N}$ and the number of vertices $v$ in $G_e$ such that $\ell_e(v) = i$ and $c(v) = a$, and
	\item $N[i,a]$ represents the minimum between $\mathcal{N}$ and the number of vertices $u \in V(G)$ with $c(u) = a$ that are upcoming neighbors with respect to $e$ of every vertex $v$ with $\ell_e(v) = i$.
\end{itemize}

For the next four lemmas, we will assume that we are given matrices $C$ and $N$ in $\intInterval{0}{\mathcal{N}}^{\NumberOfLabels \times \NumberOfColors}$.
We will describe the recursive computation of $\lambda(e,C,N)$ by distinguishing four cases depending on the kind of clique-width operation at the root of the tree $T_e$.

\begin{lemma}[Creating new vertex: $i(v)$]\label{lemma:Label}
If there exists $a \in L_v$ such that $C[i, a] = 1$ and $C[j, b] = 0$ for all the other entries $[j, b]$ in $C$, and if $check(v, a, N[i, a_1], \ldots, N[i, a_{\NumberOfColors}])$ is true, then
$\lambda(i(v), C, N) = \textsc{w}_{v,a}$.
Otherwise, $\lambda(i(v), C, N) = \Error$.

\end{lemma}

\begin{lemma}[Disjoint union: $e_1 \oplus e_2$]\label{lemma:DisjUnion}
Let $N_1$ and $N_2$ be two matrices in $\intInterval{0}{\mathcal{N}}^{\NumberOfLabels \times \NumberOfColors}$ such that:
\begin{itemize}
\item $N_1[i, a] = 0$ for every label $i \in \notFinalLabels{e_1}$ and every color $a \in \ColorSet$;
\item $N_1[i, a] = N[i, a]$ for every label $i \notin \notFinalLabels{e_1}$ and every color $a \in \ColorSet$; and
\item $N_2$ is defined analogously with respect to $e_2$.
\end{itemize}

Then
\begin{align*}
\lambda(e_1 \oplus e_2, C, N) =
	&\min\{
		\lambda(e_1, C_1, N_1) \weightsSum \lambda(e_2, C_2, N_2)
		:\\
		& \text{(a)}\; C_1, C_2 \in \intInterval{0}{\mathcal{N}}^{\NumberOfLabels \times \NumberOfColors}\\
		& \text{(b)}\; C_1[i, a] = 0 \text{ for all } i \in \notFinalLabels{e_1}, a \in \ColorSet;\\
		& \text{(c)}\; C_2[i, a] = 0 \text{ for all } i \in \notFinalLabels{e_2}, a \in \ColorSet;\\
		& \text{(d)}\; C[i,a] = \min(\mathcal{N}, C_1[i, a] + C_2[i, a]) \text{ for all } i \in \intInterval{1}{\NumberOfLabels}, a \in \ColorSet
	\}
\end{align*}

\end{lemma}

\begin{lemma}[Join: $\eta_{i,j}(e)$]\label{lemma:Join}
Let $N_e \in \intInterval{0}{\mathcal{N}}^{\NumberOfLabels \times \NumberOfColors}$ be such that
\begin{itemize}
\item $N_e[i, a] = \min(\mathcal{N}, N[i, a] + C[j, a])$ for every $a \in \ColorSet$;
\item $N_e[j, a] = \min(\mathcal{N}, N[j, a] + C[i, a])$ for every $a \in \ColorSet$;
\item $N_e[h, a] = N[h, a]$ for every $h \in \intInterval{1}{\NumberOfLabels} \setminus \{i,j\}$ and every $a \in \ColorSet$.
\end{itemize}
Then,
$\lambda(\eta_{i,j}(e), C, N) = \lambda(e, C, N_e)$.

\end{lemma}

\begin{lemma}[Relabeling: $\rho_{i \rightarrow j}(e)$]\label{lemma:Relabeling}
Let $N_e$ be such that
\begin{itemize}
\item $N_e[i, a] = N[j, a]$ for every $a \in \ColorSet$;
\item $N_e[h, a] = N[h, a]$ for every $h \in \intInterval{1}{\NumberOfLabels} \setminus \{i\}$ and every $a \in \ColorSet$.
\end{itemize}

If $C[i, a] = 0$ for all $a \in \ColorSet$, then
\begin{align*}
\lambda(\rho_{i \rightarrow j}(e), C, N) =
	\min\{ &
		\lambda(e, C_e, N_e)
		:\\
		& \text{(a)}\; C_e \in \intInterval{0}{\mathcal{N}}^{\NumberOfLabels \times \NumberOfColors}\\
		& \text{(b)}\; C[j,a] = \min(\mathcal{N}, C_e[i, a] + C_e[j, a]) \text{ for all } a \in \ColorSet;\\
		& \text{(c)}\; C_e[h, a] = C[h, a] \text{ for all } h \in \intInterval{1}{\NumberOfLabels} \setminus \{i,j\}, a \in \ColorSet
	\}.
\end{align*}
Otherwise, $\lambda(\rho_{i \rightarrow j}(e), C, N) = \Error$.

\end{lemma}

\begin{algorithm}[t]
\caption{}\label{alg:mainAlgorithm}

\BlankLine

\For{every subexpression $e$ of $e_G$, traversing them in a bottom-up fashion,}{
	\ForAll{matrices $C, N \in \intInterval{0}{\mathcal{N}}^{\NumberOfLabels \times \NumberOfColors}$}{
		Compute $\lambda(e, C, N)$ using Lemmas~\ref{lemma:Label},~\ref{lemma:DisjUnion},~\ref{lemma:Join} or~\ref{lemma:Relabeling}, according to the type of the operation at the root of $T_e$.

		Store the result in memory for future uses.
	}
}

Let $N_0$ be the matrix in $\intInterval{0}{\mathcal{N}}^{\NumberOfLabels \times \NumberOfColors}$ such that all its elements are $0$.

Let $m \leftarrow \Error$

\ForAll{$C \in \intInterval{0}{\mathcal{N}}^{\NumberOfLabels \times \NumberOfColors}$ such that $C[i, a] = 0$ for every $i \in \overline{\ell(e_G)}, a \in \ColorSet$}{
	$m \leftarrow \min(m, \lambda(e_G, C, N_0))$
}

\Return $m$
\end{algorithm}


Our algorithm, which takes the same input as a locally checkable problem, plus the number $\mathcal{N}$ and an irredundant clique-width $\NumberOfLabels$-expression $e_G$ of the input graph $G$ together with its binary rooted tree $T_{e_G}$, and outputs the minimum weight of a proper coloring of $G$, is presented in Algorithm~\ref{alg:mainAlgorithm}.
As explained above, we proceed in a bottom-up fashion, i.e. we start with the leaf nodes of $T_{e_G}$, then continue with their parents and so on, and compute each time $\lambda(e, C, N)$ for the corresponding subexpression $e$ (i.e. for the subexpression $e$ corresponding to the node of $T_{e_G}$ that we are currently analyzing) and all possible choices of $C$ and $N$ using the recurrences in Lemmas \ref{lemma:Label}, \ref{lemma:DisjUnion}, \ref{lemma:Join} and \ref{lemma:Relabeling} (see lines 1-3).
Since we are storing the results (see line 4), the number of times we need to compute some value $\lambda(\cdot, \cdot, \cdot)$ is given by the number of subexpressions of $e_G$ times the possible choices for the matrices $C$ and $N$.
Since we have $O(|V(G)|+|E(G)|)$ subexpressions in the given clique-width expression (see Section~\ref{sec:Preliminaries}), and since there exist $(\mathcal{N}+1)^{\NumberOfLabels \NumberOfColors}$ possible matrices $C$, respectively possible matrices $N$, we obtain that line 3 of our algorithm is called at most $O((|V(G)|+|E(G)|) (\mathcal{N}+1)^{2 \NumberOfLabels \NumberOfColors})$ times.
In lines 7-11, we then determine the minimum among all $\lambda(e_G, C, N_0)$, where $N_0$ is the matrix whose elements are all $0$, and $C \in \intInterval{0}{\mathcal{N}}^{\NumberOfLabels \times \NumberOfColors}$ is any matrix such that $C[i, a] = 0$ for every $i\in \overline{\ell(e_G)}$ and every $a \in \ColorSet$.
This can be done in time $O((\mathcal{N}+1)^{\NumberOfLabels \NumberOfColors})$.

It remains to determine the complexity of computing some value $\lambda(\cdot, \cdot, \cdot)$.
This clearly depends on the operation we consider.
Thus, we distinguish 4 cases:
\begin{itemize}
\item \textbf{Creating new vertex:}
We need to go through the entries of $C$, which can be done in time $O(\NumberOfLabels\NumberOfColors)$.
Let us denote by $\tcheck$ the complexity of evaluating the check function.
Hence, we obtain a complexity of $O(\NumberOfLabels\NumberOfColors + \tcheck)$ for this operation.

\item \textbf{Disjoint union:}
We first need to determine $N_1$ and $N_2$, which takes $O(\NumberOfLabels\NumberOfColors)$ time, and then we need to find the minimum weight by going through all possible choices of $C_1$ and $C_2$, which can be done in time $O((\mathcal{N}+1)^{2 \NumberOfLabels \NumberOfColors})$.
This gives us an overall complexity of $O((\mathcal{N}+1)^{2 \NumberOfLabels\NumberOfColors})$ for determining $\lambda(\cdot, \cdot, \cdot)$ for the disjoint union operation.

\item \textbf{Join:}
We simply need to determine the matrix $N_e$, which can be done in $O(\NumberOfLabels\NumberOfColors)$ time.

\item \textbf{Relabeling:}
First, we need to determine the matrix $N_e$, which takes $O(\NumberOfLabels\NumberOfColors)$ time, and then we need to find the minimum weight by considering possible choices of $C_e$ with all rows fixed except two, which clearly takes $O((\mathcal{N}+1)^{2 \NumberOfColors})$.
Thus, overall the complexity of determining $\lambda(\cdot, \cdot, \cdot)$ for the relabeling operation is $O(\NumberOfLabels\NumberOfColors + (\mathcal{N}+1)^{2 \NumberOfColors})$.
\end{itemize}

Now the complexity of computing any $\lambda(e, C, N)$ is bounded by the sum of the complexities of the four cases, for which we obtain $O(\tcheck + (\mathcal{N}+1)^{2 \NumberOfLabels \NumberOfColors})$.
Thus, we obtain the following complexity:
\begin{center}
	$O(
		(|V(G)|+|E(G)|) (\mathcal{N}+1)^{2 \NumberOfLabels \NumberOfColors}
		(\tcheck + (\mathcal{N}+1)^{2 \NumberOfLabels \NumberOfColors})
	).$
\end{center}

\begin{remark}
We can modify the algorithm in order to also obtain the coloring function as an output.
This does not affect the complexity.
\end{remark}

Let us highlight the following main consequences of the previous analysis.
Notice that, by the results in~\cite{approx-cw}, we do not need a clique-width expression as input.

\begin{corollary}
\label{cor:xp}
Consider a color-counting 1-locally checkable problem $\Pi$ with constant number of colors and a check function computable in polynomial time.
Then $\Pi$ is XP parameterized by clique-width.
\end{corollary}

\begin{corollary}
\label{cor:d-stable}
Let $d \in \mathbb{N}$.
If $\Pi$ is a $d$-stable 1-locally checkable problem where the number of colors is $O(\log |V(G)|)$ and the check function can be computed in polynomial time, then $\Pi$ is XP parameterized by clique-width.
\end{corollary}

\begin{corollary}
\label{cor:linear}
Let $d \in \mathbb{N}$.
If $\Pi$ is a $d$-stable 1-locally checkable problem where the number of colors is constant and the check function can be computed in constant time, then $\Pi$ is FPT parameterized by clique-width.
Moreover, if an irredundant clique-width $k$-expression is given as input, then it is linear FPT parameterized by $k$.
\end{corollary}

Notice that various well-known graph theoretical problems, such as {\sc $k$-Coloring}, {\sc Maximum Independent Set}, as well as {\sc $[k]-$Roman domination} (see Section~\ref{sec:Applications}), are indeed $d$-stable 1-locally checkable problems, for some constant $d$, with constant number of colors.

\section{Global size property}\label{sec:SizeGlobProp}

In this section, we extend the results of Section~\ref{sec:LCP} by considering color-counting 1-locally checkable problems in which it is also required that the number of vertices that receive a given color $a\in \ColorSet$ belongs to a predefined set $\sigma_a$ of non-negative integers.

Let $(Q, \{1\}, \delta, q_0, F)$ be a deterministic finite-state automaton which accepts a string of $t$ consecutive $1$'s if and only if $t \in \sigma_a$.
Note that for all finite sets of non-negative integers, there exists such an automaton (for example, let $m$ be the maximum element of the set, then we set $Q = \{s_0, \ldots, s_{m+1}\}$, $q_0 = s_0$, $F = \setst{s_i}{i \in \sigma}$, $\delta(s_i, 1) = s_{i+1}$ for all $0 \leq i \leq m$ and $\delta(s_{m+1}, 1) = s_{m+1}$).
Let us define the notation $\delta^0(s_i) = s_i$ and $\delta^n(s_i) = \delta(\delta^{n-1}(s_i), 1)$ for every state $s_i \in Q$ and positive integer $n$.

We will now proceed in a similar way as in Section~\ref{sec:Algorithm} but considering additional parameters.
Let us first introduce the relevant notion of \emph{$(C, N, p_1, \ldots, p_m)$-colorings}, which will be defined recursively.
This notion can be used to extend the results of the aforementioned section using different global properties.
Intuitively, if $C, N \in \intInterval{0}{\mathcal{N}}^{\NumberOfLabels \times \NumberOfColors}$ and $p_1, \ldots, p_m$ are parameters such that $(C, N, p_1, \ldots, p_m)$-colorings of $G_e$ are defined, then, for additional parameters $p_{m+1},\ldots,p_{m+m'}$, we define a \emph{$(C,N,p_1,\ldots,p_m,p_{m+1},\ldots,p_{m+m'})$-coloring of $G_e$} as a $(C,N,p_1,\ldots,p_m)$-coloring of $G_e$ such that parameters $p_{m+1},\ldots,p_{m+m'}$ satisfy some predefined property.
In the case of the particular global property mentioned at the beginning of this section, we will only consider two additional parameters.
The first such parameter is a state $s_a \in Q$ and the second parameter is a function $f_a \colon Q \rightarrow \BoolSet$.
We then define a \emph{$(C,N,p_1,\ldots,p_m, s_a,f_a)$-coloring} $c$ of $G_e$ as a $(C,N,p_1,\ldots,p_m)$-coloring of $G_e$ such that $f_a(\delta^{n}(s_a)) = \True$, where $n = |\setst{v \in V(G_e)}{c(v) = a}|$.
Also, in the same spirit as before, we will denote by $\lambda(e, C, N, p_1, \ldots, p_m, s_a, f_a)$ the minimum weight among all $(C, N, p_1, \ldots, p_m, s_a, f_a)$-colorings of $G_e$.
If we want to fix the size of $\mathcal{R}$ color classes, say $a_1, \ldots, a_{\mathcal{R}}$, it suffices to associate an automaton $M_i$ and the corresponding parameters $s_{a_i}$ and $f_{a_i}$ with each color class $a_i$, for $i \in \intInterval{1}{\mathcal{R}}$.

By providing a lemma explaining how to solve a color-counting 1-locally checkable problem with given global properties by using $(C, N, p_1, \ldots, p_m, s_a, f_a)$-colorings, and then again distinguishing the four clique-width operations, as in the previous section, we can prove that, when the number of colors is constant, this new algorithm is also XP parameterized by clique-width.
Due to space restrictions, their statements are omitted here, but presented in Appendix~\ref{sec:Proofs}.

\section{Applications}\label{sec:Applications}

In this section, we provide some examples of problems whose complexity status in graphs of bounded clique-width was unknown, and for each of which the application of our framework yields a first polynomial-time algorithm in this class of graphs.

\subsection{(Global) $[k]-$Roman domination}\label{edgeDeletionChrom}

The \textsc{$[k]-$Roman domination} problem was first defined in~\cite{tripleRoman} as a generalization of Roman and double Roman domination~\cite{RomanDomination,doubleRoman}. 
Let $k \geq 1$ be an integer.
A \emph{$[k]$-Roman dominating function} on a graph $G$ is a function $f \colon V(G) \to \intInterval{0}{k+1}$
having the property that
if $f(v) < k$ then
$\sum_{u \in N_G[v]} f(u) \geq |AN_G^f(v)| + k$,
where $AN_G^f(v) = \{u \in N_G(v) : f(u) \geq 1\}$ (this set is called the \emph{active neighborhood of $v$}).
The \emph{weight} of a $[k]$-Roman dominating function $f$ is $\sum_{v \in V(G)} f(v)$, and the minimum weight of a $[k]$-Roman dominating function on $G$ is the \emph{$[k]$-Roman domination number} of $G$, denoted by $\gamma_{[kR]}(G)$.
The problem consists in computing the $[k]$-Roman domination number of a given graph.

In~\cite{LCPTreewidth}, this problem was shown to be solvable in linear time in graphs of bounded treewidth.
In their model, the number of colors is a constant and the check function is actually $(k+1)$-stable.
We can express it in the following way:
\begin{itemize}
\item $\ColorSet=\intInterval{0}{k+1}$
and $L_{v} = \intInterval{0}{k+1}$ for all $v \in V(G)$;

\item $(\WeightSet, \wlesseq, \weightsSum)=(\mathbb{N}\cup \{+\infty\},\leq, +)$
and $\textsc{w}_{v,a} = a$ for all $v \in V(G), a \in L_{v}$;

\item $check(v, a, \upcomingNeighbors_0, \ldots, \upcomingNeighbors_{k+1}) = \left(a + \sum_{j=0}^{k+1} j \upcomingNeighbors_j \geq k + \sum_{j=1}^{k+1} \upcomingNeighbors_j\right)$.
\end{itemize}

Then, by Corollary \ref{cor:linear}, this problem is FPT parameterized by clique-width (and linear FPT when a suitable clique-width expression is given).

In~\cite{GlobRomIntrod}, the authors introduced a variant of this problem, called \textsc{global Roman domination}.
This problem was later extended to \textsc{global double Roman domination}~\cite{GlobDoubleRomDom} and \textsc{global triple Roman domination}~\cite{GlobalKRomDom}.
The definition of these problems can be naturally generalized as follows.
A \emph{global $[k]-$Roman dominating function} on a graph $G$ is a $[k]-$Roman dominating function in both $G$ and $\overline{G}$.
The \textsc{global $[k]-$Roman domination} problem consists in computing the minimum weight of a global $[k]-$Roman dominating function of a graph.

In order to show that this problem is XP parameterized by clique-width, we first define an auxiliary problem.

\medskip

\defProbEnglish
{Specified size global $k-$Roman domination}
{A graph $G$ and $k+2$ non-negative integers $\SizeOfColorClass_0, \ldots, \SizeOfColorClass_{k+1}$ such that $\sum_{i=0}^{k+1} \SizeOfColorClass_i = |V(G)|$.}
{Does $G$ admit a global $[k]-$Roman dominating function $f$ such that, for all $i \in \intInterval{0}{k+1}$, $\SizeOfColorClass_i$ equals the number of vertices $v \in V(G)$ with $f(v) = i$?}

This last problem can be modeled as a color-counting 1-locally checkable problem with global properties:
\begin{itemize}
\item $\ColorSet = \intInterval{0}{k+1}$
and $L_{v} = \intInterval{0}{k+1}$ for all $v \in V(G)$;

\item $(\WeightSet, \wlesseq, \weightsSum) = (\mathbb{N}\cup \{+\infty\},\leq, +)$
and $\textsc{w}_{v, a} = a$ for all $v \in V(G), a \in L_{v}$;

\item $check(v, a, \upcomingNeighbors_0, \ldots, \upcomingNeighbors_{k+1}) = (a + \sum_{j=1}^{k+1}(j-1)\upcomingNeighbors_j \geq k) \land (\sum_{j=1}^{k+1}(j-1)(\SizeOfColorClass_j - \upcomingNeighbors_j) \geq k)$; 

\item for all $a \in \ColorSet$, we ask for the size of the color class of $a$ to belong to $\{\SizeOfColorClass_a\}$.
\end{itemize}

Finally, to solve \textsc{global $[k]-$Roman domination} on graphs of bounded clique-width, we successively iterate over the feasible combinations of values $\SizeOfColorClass_0, \ldots, \SizeOfColorClass_{k+1}$ such that $\sum_{i=0}^{k+1} \SizeOfColorClass_i = |V(G)|$ and $\SizeOfColorClass_i \geq 0$ for all $i \in \intInterval{0}{k+1}$.
Notice that the number of such combinations is no more than $(|V(G)|+1)^{k+2}$.
For each combination, we solve \textsc{Specified size global $k-$Roman domination}, and we retain the solution of minimum weight.


\subsection{$k$-community, Max PDS and other variants}\label{k-community}

The notion of \textit{community structure} was first introduced in \cite{communityIntro}, as a partition $\Pi = \{C_1, \ldots, C_k\}$, with $k \geq 2$, of the set of vertices of a graph into so called \textit{communities}, such that for each $i \in \intInterval{1}{k}$ we have $|C_i| \geq 2$ and, for each vertex $v \in C_i$ and each community $C_j \neq C_i$, $\frac{|N_G(v) \cap C_i|}{|C_i|-1}\geq \frac{|N_G(v) \cap C_j|}{|C_j|}\,.$
Finding a community structure in any graph $G$ can be done in polynomial time (see~\cite{communityIntro}).
However, the number of communities $k$ in the obtained community structure can be any value between 2 and $\frac{|V(G)|}{2}$, and the algorithm does not apply when we want to impose the number of communities.
The \textsc{2-community} problem was introduced in~\cite{k-Community} as the problem of deciding whether a given connected graph has a \textit{2-community structure}, i.e. a community structure with 2 communities.
This can be naturally generalized to the \textsc{$k$-community} problem, for any fixed $k$, as the problem of deciding whether a given connected graph has a community structure with $k$ communities.
The complexity status of \textsc{2-community} is still unknown, and only a few graph classes are known to admit polynomial time algorithms for this problem (for instance, graphs of maximum degree 3 and graphs of minimum degree $|V(G)| - 3$~\cite{k-Community}).

We show here that \textsc{$k$-community} is XP parameterized by clique-width.
Our approach is similar to the one for \textsc{global $[k]-$Roman domination}, in the sense that we define a variant of the problem where we require a certain size of each community, to which we reduce \textsc{$k$-community}.

\medskip

\defProbEnglish
	{Specified size $k$-community}
	{A graph $G$ and $k$ integers $\SizeOfColorClass_1, \ldots, \SizeOfColorClass_k \geq 2$, such that $\sum_{i=1}^k \SizeOfColorClass_i = |V(G)|$.}
	{Does $G$ admit a $k$-community structure $\Pi = \{C_1, \ldots, C_k\}$ such that $\vert C_i \vert = \SizeOfColorClass_i$ for all $i \in \intInterval{1}{k}$?}

The \textsc{Specified size $k$-community} problem can be modeled as a color-counting 1-locally checkable problem with global properties.
Notice that since it is a decision problem, we only need two values for the weight set.
\begin{itemize}
\item $\ColorSet = \intInterval{1}{k}$
and $L_{v} = \intInterval{1}{k}$ for all $v \in V(G)$;

\item $(\WeightSet, \wlesseq, \weightsSum) = (\{0,1\}, \leq, \max)$
and $\textsc{w}_{v, a} = 0$ for all $v \in V(G), a \in L_v$;

\item $check(v, a, \upcomingNeighbors_1,\ldots, \upcomingNeighbors_{\NumberOfColors}) = \left(\forallFormula{b \in \intInterval{1}{k}}{\frac{\upcomingNeighbors_{a}}{\SizeOfColorClass_{a} - 1} \geq \frac{\upcomingNeighbors_b}{\SizeOfColorClass_b}}\right)$;

\item for all $a \in \ColorSet$, we ask for the size of the color class of $a$ to belong to $\{\SizeOfColorClass_a\}$.
\end{itemize}

Then, \textsc{$k$-community} can be solved by successively iterating over the feasible combinations of values $\SizeOfColorClass_1, \ldots, \SizeOfColorClass_k$ such that $\sum_{i=1}^{k} \SizeOfColorClass_i = |V(G)|$ and $\SizeOfColorClass_i \geq 2$ for all $i \in \intInterval{1}{k}$, and for each of the combinations solving \textsc{Specified size $k$-community}.

Note that \textsc{Balanced $k$-community}, i.e. the problem of finding a $k$-community structure with all parts having the same size, is equivalent to \textsc{Specified size $k$-community} with $\SizeOfColorClass_i = \SizeOfColorClass_j$, for all $i,j \in \intInterval{1}{k}$.
Hence, it is also XP parameterized by clique-width.
In~\cite{balK-comNPC}, it was shown that this problem is NP-complete in general, and in~\cite{k-Community} it was pointed out to be polynomially solvable in graphs of bounded treewidth.
It is not difficult to see that the problem \textsc{Weak $k$-community}, defined in~\cite{k-Community}, can also be solved by slightly modifying the above check function.

A closely related problem is the \textsc{Maximum Proportionally Dense Subgraph (Max PDS)} problem, originally defined in~\cite{bazgan-proportionally-2019}.
Let $G$ be a graph and $S \subset V(G)$, such that $2 \leq |S| < |V(G)|$.
We say that the induced subgraph $G[S]$ is a \emph{proportionally dense subgraph (PDS)} if for every $v\in S$, we have
$\frac{|N_G(v) \cap S|}{|S|-1}\geq \frac{|N_G(v) \cap \overline{S}|}{|\overline{S}|}$.
Then, the \textsc{Max PDS} problem consists in finding a proportionally dense subgraph in $G$ of maximum size.
The authors of~\cite{bazgan-proportionally-2019} showed that the \textsc{Max PDS} problem is NP-hard, even when restricted to split graphs or bipartite graphs, and that it can be solved in linear time in cubic Hamiltonian graphs.

By proceeding in a similar way as before, where in the associated auxiliary problem we have only two colors, $\textsc{s}$ and $\overline{\textsc{s}}$, and the check function is given by $check(v, a, \upcomingNeighbors_0, \upcomingNeighbors_1) = \left(a = \textsc{s} \Rightarrow {\frac{\upcomingNeighbors_{\textsc{s}}}{\SizeOfColorClass_{\textsc{s}} - 1} \geq \frac{\upcomingNeighbors_{\overline{\textsc{s}}}}{\SizeOfColorClass_{\overline{\textsc{s}}}}}\right)$,
we can show that \textsc{Max PDS} is XP parameterized by clique-width.

Another variation defined in~\cite{bazgan-proportionally-2019} is the \textsc{PDS Extension} problem, which asks whether there exists a proportionally dense subgraph $G[S]$ such that $U \subset S$, for some $U \subset V(G)$ given as an input.
It was shown in~\cite{bazgan-proportionally-2019} that the \textsc{PDS Extension} problem is NP-complete, and no polynomial time algorithms were known for any graph class.
We can show that this problem is also XP parameterized by clique-width, by proceeding almost exactly as explained above, where the only change is that we now set $L_v = \{\textsc{s}\}$ for all $v \in U$.

Given a graph $G$ and a real number $\gamma \in (0, 1]$, a \emph{degree-based $\gamma$-quasi-clique} is defined as a subset $S \subseteq V(G)$ such that the degree of any vertex in $G[S]$ is at least $\gamma (|S| - 1)$, that is, $\frac{|N_G(v) \cap S|}{|S|-1}\geq \gamma$.
The \textsc{maximum degree-based $\gamma$-quasi-clique} problem consists in finding a degree-based $\gamma$-quasi-clique of maximum cardinality in a graph.
In~\cite{max-degree-based-quasi-clique}, it was shown that this problem is NP-hard for any fixed $\gamma$.
Using the same techniques as for \textsc{Max-PDS} (only slightly modifying the check function), we obtain that \textsc{maximum degree-based $\gamma$-quasi-clique} is XP parameterized by clique-width.

\bibliographystyle{abbrv}
\bibliography{Z-References}

\newpage
\appendix
\gdef\thesection{\Alph{section}}
\makeatletter
\renewcommand\@seccntformat[1]{Appendix \csname the#1\endcsname.\hspace{0.5em}}
\makeatother
\section{Omitted proofs and examples}
\label{sec:Proofs}

In this section, we include the proofs of the lemmas presented in Section~\ref{sec:Algorithm}, as well as the lemmas omitted from~\ref{sec:SizeGlobProp} and their proofs.
We also include examples of well known problems modeled as locally checkable problems.

\subsection{Examples of color-counting 1-locally checkable problems}

\begin{example}
Consider the {\sc $k$-Coloring} problem.
This problem can be seen as a color-counting $1$-locally checkable problem with the following characteristics:
\begin{itemize}
	\item $\ColorSet = \intInterval{1}{k}$
	and $L_{v} = \intInterval{1}{k}$ for all $v \in V(G)$;

	\item $(\WeightSet, \wlesseq, \weightsSum) = (\{0,1\}, \leq, \max)$
	and $\textsc{w}_{v, a} = 0$ for all $v \in V(G), a \in L_v$;

	\item $check(v, a, \upcomingNeighbors_1, \ldots, \upcomingNeighbors_{k}) = (\upcomingNeighbors_a = 0)$.
\end{itemize}
\end{example}

\begin{example}
The {\sc Maximum Independent Set} problem can also be modeled as a color-counting $1$-locally checkable problem:
\begin{itemize}
	\item $\ColorSet = \{0,1\}$
	and $L_{v} = \{0,1\}$ for all $v \in V(G)$;

	\item ($\WeightSet, \wlesseq, \weightsSum) = (\mathbb{N}\cup \{-\infty\}, \geq, +)$
	and $\textsc{w}_{v, a} = a$ for all $v \in V(G), a \in L_v$;

	\item $check(v, a, \upcomingNeighbors_0, \upcomingNeighbors_1) = (a = 0 \lor \upcomingNeighbors_1 = 0)$.
\end{itemize}
\end{example}

\begin{example}
The {\sc Minimum Odd Dominating Set} problem can as well be modeled as a color-counting $1$-locally checkable problem, as follows:
\begin{itemize}
	\item $\ColorSet = \{0,1\}$
	and $L_{v} = \{0,1\}$ for all $v \in V(G)$;

	\item ($\WeightSet, \wlesseq, \weightsSum) = (\mathbb{N}\cup \{+\infty\}, \leq, +)$
	and $\textsc{w}_{v, a} = a$ for all $v \in V(G), a \in L_v$;

	\item $check(v, a, \upcomingNeighbors_0, \upcomingNeighbors_1) = (a + \upcomingNeighbors_1 \equiv 1 \mod 2)$.
\end{itemize}
\end{example}

Notice that the check functions of {\sc $k$-Coloring} and {\sc Maximum Independent Set} are both 1-stable, but the check function of {\sc Minimum Odd Dominating Set} is not $d$-stable for any constant $d$.

\subsection{Omitted proofs of Section~\ref{sec:Algorithm}}


\noindent\textbf{Lemma \ref{lem:initialLemma}.}
\emph{Let $\Pi$ be a color-counting 1-locally checkable problem with input graph $G$ and let $e_G$ be a clique-width $\NumberOfLabels$-expression of $G$.
Then the minimum weight of a proper coloring of $G$ equals the minimum among all $\lambda(e_G, C, N_0)$, where $N_0 \in \intInterval{0}{\mathcal{N}}^{\NumberOfLabels \times \NumberOfColors}$ is the matrix whose elements are all $0$ and $C \in \intInterval{0}{\mathcal{N}}^{\NumberOfLabels \times \NumberOfColors}$ is any matrix such that $C[i, a] = 0$ for every $i\in \overline{\ell(e_G)}$ and every $a \in \ColorSet$.
}
\begin{proof}
We will show that for every proper coloring $c$ of $G$ there exists a matrix $C \in \intInterval{0}{\mathcal{N}}^{\NumberOfLabels \times \NumberOfColors}$ such that $C[i, a] = 0$ for every $i \in \overline{\ell(e_G)}$ and every $a \in \ColorSet$, and such that $\textsc{w}(c) \geq \lambda(e_G, C, N_0)$.
On the other hand, we will then show that for every matrix $C \in \intInterval{0}{\mathcal{N}}^{\NumberOfLabels \times \NumberOfColors}$ such that $C[i, a] = 0$ for every $i\in \overline{\ell(e_G)}$ and every $a \in \ColorSet$, and such that $\lambda(e_G, C, N_0) \neq \Error$, there exists a proper coloring $c$ of $G$ such that $\textsc{w}(c) = \lambda(e_G, C, N_0)$.

Suppose we have a proper coloring $c$ of $G$.
Let $C \in \intInterval{0}{\mathcal{N}}^{\NumberOfLabels \times \NumberOfColors}$ be the matrix such that $C[i,a] = \min(\mathcal{N}, |\setst{v \in V(G)}{c(v) = a \land \ell_e(v) = i}|)$ for all $i \in \intInterval{1}{\NumberOfLabels}$ and all $a \in \ColorSet$.
Clearly, $C[i, a] = 0$ for every $i \in \overline{\ell(e_G)}$ and every $a \in \ColorSet$.
Also, for every $v \in V(G)$ we have that $check(v, c(v), \upcomingNeighbors_1, \ldots, \upcomingNeighbors_{\NumberOfColors})$ is true, where $\upcomingNeighbors_j = \min(\mathcal{N}, N_0[\ell_e(v), a_j] + |\setst{u \in N_{G}(v)}{c(u) = a_j}|)$ for every $j \in \intInterval{1}{\NumberOfColors}$, because $N_0[\ell_e(v), a_j] = 0$ for every $j \in \intInterval{1}{\NumberOfColors}$ and $c$ is a proper coloring of $G$.
Therefore, $c$ is a $(C, N_0)$-coloring of $G$ and so $\textsc{w}(c) \geq \lambda(e_G, C, N_0)$.

Now suppose we have a matrix $C \in \intInterval{0}{\mathcal{N}}^{\NumberOfLabels \times \NumberOfColors}$ such that $C[i, a] = 0$ for every $i\in \overline{\ell(e_G)}$ and every $a \in \ColorSet$, and such that $\lambda(e_G, C, N_0) \neq \Error$. 
Let $c$ be a $(C, N_0)$-coloring of $G$ of minimum weight (notice that at least one such $c$ exists, since $\lambda(e_G, C, N_0) \neq \Error$). 
We will prove that $c$ is a proper coloring of $G$.
By definition of a $(C, N_0)$-coloring, $c$ is a valid coloring, so it only remains to prove that $check(G, v, c|_{N_G[v]})$ is true for every $v \in V(G)$.
We know that for every $v \in V(G)$, we have that $check(v, c(v), \upcomingNeighbors_1, \ldots, \upcomingNeighbors_{\NumberOfColors})$ is true, where $\upcomingNeighbors_j = \min(\mathcal{N}, N_0[\ell_e(v), a_j] + |\setst{u \in N_{G}(v)}{c(u) = a_j}|)$ for every $j \in \intInterval{1}{\NumberOfColors}$.
Since $N_0[\ell_e(v), a_j] = 0$ for every $v \in V(G)$ and $j \in \intInterval{1}{\NumberOfColors}$, we have that $check(G, v, c|_{N_G[v]}) = check(v, c(v), \upcomingNeighbors_1, \ldots, \upcomingNeighbors_{\NumberOfColors}) = \True$, where $\upcomingNeighbors_j = \min(\mathcal{N}, |\setst{u \in N_G(v)}{c(u) = a_j}|)$ for all $j \in \intInterval{1}{\NumberOfColors}$.
\end{proof}

\noindent\textbf{Lemma \ref{lemma:Label} (Creating new vertex: $i(v)$).}
\emph{If there exists $a \in L_v$ such that $C[i, a] = 1$ and $C[j, b] = 0$ for all the other entries $[j, b]$ in $C$, and if $check(v, a, N[i, a_1], \ldots, N[i, a_{\NumberOfColors}])$ is true, then
$\lambda(i(v), C, N) = \textsc{w}_{v,a}$.
Otherwise, $\lambda(i(v), C, N) = \Error$.
}
\begin{proof}
First notice that $G_{i(v)}$ is the graph consisting of a single vertex $v$ with label $i$.
Therefore, if $C$ and $N$ have the above properties (i.e. there exists $a \in L_v$ such that $C[i, a] = 1$ and $C[j, b] = 0$ for all the other entries $[j, b]$ in $C$, and $check(v, a, N[i, a_1], \ldots, N[i, a_{\NumberOfColors}])$ is true), there is exactly one $(C,N)$-coloring $c$ of $G_{i(v)}$, defined by $c(v) = a$.
Indeed, since $a \in L_v$, it follows that $c$ is a valid coloring.
Moreover, conditions (C1) and (C2) are trivially satisfied.
Then, $\lambda(i(v), C, N) = \textsc{w}(c) = \textsc{w}_{v,a}$.

On the other hand, if $C$ does not have exactly one nonzero entry, or if it is not in row $i$, or if it is in a column $a\not\in L_v$, or if this entry is not equal to 1, then no valid coloring satisfying condition (C1) exists.
If for this unique possible choice of color $a$, $check(v, a, N[i, a_1], \ldots, N[i, a_{\NumberOfColors}])$ is false, then no $(C,N)$-coloring of $G_{i(v)}$ exists either.
Therefore $\lambda(i(v), C, N) = \Error$.
\end{proof}

\noindent\textbf{Lemma \ref{lemma:DisjUnion} (Disjoint union: $e_1 \oplus e_2$).}
\emph{Let $N_1$ and $N_2$ be two matrices in $\intInterval{0}{\mathcal{N}}^{\NumberOfLabels \times \NumberOfColors}$ such that:
\begin{itemize}
\item $N_1[i, a] = 0$ for every label $i \in \notFinalLabels{e_1}$ and every color $a \in \ColorSet$;
\item $N_1[i, a] = N[i, a]$ for every label $i \notin \notFinalLabels{e_1}$ and every color $a \in \ColorSet$; and
\item $N_2$ is defined analogously with respect to $e_2$.
\end{itemize}

Then
\begin{align*}
\lambda(e_1 \oplus e_2, C, N) =
	&\min\{
		\lambda(e_1, C_1, N_1) \weightsSum \lambda(e_2, C_2, N_2)
		:\\
		& \text{(a)}\; C_1, C_2 \in \intInterval{0}{\mathcal{N}}^{\NumberOfLabels \times \NumberOfColors}\\
		& \text{(b)}\; C_1[i, a] = 0 \text{ for all } i \in \notFinalLabels{e_1}, a \in \ColorSet;\\
		& \text{(c)}\; C_2[i, a] = 0 \text{ for all } i \in \notFinalLabels{e_2}, a \in \ColorSet;\\
		& \text{(d)}\; C[i,a] = \min(\mathcal{N}, C_1[i, a] + C_2[i, a]) \text{ for all } i \in \intInterval{1}{\NumberOfLabels}, a \in \ColorSet
	\}
\end{align*}
}
\begin{proof}
Let $\alpha=\min\{\lambda(e_1, C_1, N_1) \weightsSum \lambda(e_2, C_2, N_2) : (a), (b), (c), (d)$ are satisfied$\}$.
We will first prove that $\lambda(e_1 \oplus e_2, C, N) \geq \alpha$.
If $\lambda(e_1 \oplus e_2, C, N) = \Error$, then we are done.
So assume that $\lambda(e_1 \oplus e_2, C, N) \neq \Error$ and let $c$ be a $(C,N)$-coloring of $G_{e_1 \oplus e_2}$ whose weight equals $\lambda(e_1 \oplus e_2,C,N)$.
We need to show that there exist $C_1$ and $C_2$ in $\intInterval{0}{\mathcal{N}}^{\NumberOfLabels \times \NumberOfColors}$ satisfying $(b),(c),(d)$ and such that the weight of $c$ is at least $\lambda(e_1, C_1, N_1) \weightsSum \lambda(e_2, C_2, N_2)$.
Let $c_1 = c|_{V(G_{e_1})}$ and $c_2 = c|_{V(G_{e_2})}$. Then, we define $C_1[i,a]=\min(\mathcal{N}, |\{v \in V(G_{e_1}) : c_1(v) = a \land \ell_{e_1}(v) = i\}|)$ for any label $i\in \intInterval{1}{\NumberOfLabels}$ and color $a\in \ColorSet$, and similarly for $C_2$ with respect to $e_2$.
Consequently, conditions $(b)$ and $(c)$ are satisfied: if $i \in \notFinalLabels{e_1}$, then $\{v \in V(G_{e_1}) : \ell_{e_1}(v) = i\}=\emptyset$, thus $C_1[i,a] = 0$, and similarly for $C_2$.
Condition $(d)$ is also satisfied because $c$ is a $(C,N)$-coloring of $G_{e_1 \oplus e_2}$, $V(G_{e_1})$ and $V(G_{e_2})$ are disjoint, $\ell_{e_1}(v)=\ell_{e_1 \oplus e_2}(v)$ for every $v \in V(G_{e_1})$ and $\ell_{e_2}(v)=\ell_{e_1 \oplus e_2}(v)$ for every $v \in V(G_{e_2})$.
We now show that $c_1$ is a $(C_1, N_1)$-coloring of $G_{e_1}$.
Condition (C1) is trivially satisfied by the definition of $C_1$.
To show that condition (C2) is satisfied, we will show that $check(v, c_1(v), \upcomingNeighbors'_1, \ldots, \upcomingNeighbors'_\NumberOfColors)$ is true for every vertex $v \in V(G_{e_1})$, where $\upcomingNeighbors'_j = \min(\mathcal{N}, N_1[\ell_{e_1}(v), a_j] + |\setst{u \in N_{G_{e_1}}(v)}{c_1(u) = a_j}|)$ for every $j \in \intInterval{1}{\NumberOfColors}$.
Since $c$ is a $(C,N)$-coloring of $G_{e_1 \oplus e_2}$, we have that $check(v, c(v), \upcomingNeighbors_1, \ldots, \upcomingNeighbors_\NumberOfColors)$ is true for every $v \in V(G_{e_1 \oplus e_2})$, where $\upcomingNeighbors_j = \min(\mathcal{N}, N[\ell_{e_1 \oplus e_2}(v), a_j] + |\setst{u \in N_{G_{e_1 \oplus e_2}}(v)}{c(u) = a_j}|)$ for every $j \in \intInterval{1}{\NumberOfColors}$.
By definition of $N_1$ and since $\ell_{e_1}(v) = \ell_{e_1 \oplus e_2}(v)$ for every $v \in V(G_{e_1})$, we have $N_1[\ell_{e_1}(v), a_j]=N[\ell_{e_1 \oplus e_2}(v), a_j]$ for every $v \in V(G_{e_1})$ and every $j \in \intInterval{1}{\NumberOfColors}$.
Also, $N_{G_{e_1}}(v) = N_{G_{e_1 \oplus e_2}}(v)$ for every $v \in V(G_{e_1})$ and  $c_1 = c|_{V(G_{e_1})}$, so $|\setst{u \in N_{G_{e_1}}(v)}{c_1(u) = a_j}| = |\setst{u \in N_{G_{e_1 \oplus e_2}}(v)}{c(u) = a_j}|$ for every $v \in V(G_{e_1})$ and every $j \in \intInterval{1}{\NumberOfColors}$.
Therefore $\upcomingNeighbors'_j = \upcomingNeighbors_j$ for every $j \in \intInterval{1}{\NumberOfColors}$ and hence, $check(v, c_1(v), \upcomingNeighbors'_1, \ldots, \upcomingNeighbors'_\NumberOfColors)$ is true for every $v \in V(G_{e_1})$.
Using similar arguments, we can show that $c_2$ is a $(C_2, N_2)$-coloring of $G_{e_2}$.
Moreover, $\textsc{w}(c) = \textsc{w}(c_1) \weightsSum \textsc{w}(c_2)$ by definition of $c_1$ and $c_2$, and consequently, $\lambda(e_1 \oplus e_2, C,N) = \textsc{w}(c) = \textsc{w}(c_1) \weightsSum \textsc{w}(c_2) \geq \lambda(e_1, C_1, N_1) \weightsSum \lambda(e_2, C_2, N_2) \geq \alpha$.

Let us show now that $\lambda(e_1 \oplus e_2, C, N) \leq \alpha$.
Consider two matrices $C_1$ and $C_2$ in $\intInterval{0}{\mathcal{N}}^{\NumberOfLabels \times \NumberOfColors}$ satisfying $(b),(c),(d)$.
If $\lambda(e_1, C_1, N_1) = \Error$ or $\lambda(e_2, C_2, N_2) = \Error$, then we are done.
Otherwise, we are going to construct a $(C, N)$-coloring $c$ of $G_{e_1 \oplus e_2}$ such that $\textsc{w}(c) = \lambda(e_1, C_1, N_1) \weightsSum \lambda(e_2, C_2, N_2)$.
Let $c_1$ be a $(C_1,N_1)$-coloring of $G_{e_1}$ whose weight is $\lambda(e_1, C_1, N_1)$ and $c_2$ be a $(C_2,N_2)$-coloring of $G_{e_2}$ whose weight is $\lambda(e_2, C_2, N_2)$.
Let $c = c_1 \cup c_2$ (note that $c$ is well defined because $V(G_{e_1})$ and $V(G_{e_2})$ are disjoint sets).
Clearly, $c$ is a valid coloring and $\textsc{w}(c) = \textsc{w}(c_1) \weightsSum \textsc{w}(c_2)$.
We now show that $c$ is a $(C, N)$-coloring of $G_{e_1 \oplus e_2}$.
We start with condition (C1).
Consider $i \in \intInterval{1}{\NumberOfLabels}$ and $a \in \ColorSet$.
Since $c_1$ is a $(C_1, N_1)$-coloring of $G_{e_1}$ and $c_2$ is a $(C_2, N_2)$-coloring of $G_{e_2}$, we have $C_1[i,a] = \min(\mathcal{N}, |\{v \in V(G_{e_1}) : c_1(v) = a \land \ell_{e_1}(v) = i\}|)$ and $C_2[i,a] = \min(\mathcal{N}, |\{v \in V(G_{e_2}) : c_2(v) = a \land \ell_{e_2}(v) = i\}|))$.
Then,
\begin{align*}
C[i,a] =	&	\min(\mathcal{N}, C_1[i,a] + C_2[i,a])	\\
	=	&	\min(\mathcal{N}, \min(\mathcal{N}, |\{v \in V(G_{e_1}) : c_1(v) = a \land \ell_{e_1}(v) = i\}|) +  \\
	         &     \min(\mathcal{N}, |\{v \in V(G_{e_2}): c_2(v) = a \land \ell_{e_2}(v) = i\}|))	\\
	=	&	\min(\mathcal{N}, |\{v \in V(G_{e_1}) : c_1(v) = a \land \ell_{e_1}(v) = i\}|+\\
	        &       |\{v \in V(G_{e_2}) : c_2(v) = a \land \ell_{e_2}(v) = i\}|)	\\
	=	&	\min(\mathcal{N}, |\{v \in V(G_{e_1 \oplus e_2}) : c(v) = a \land \ell_{e_1 \oplus e_2}(v) = i\}|).
\end{align*}
The last equality is simply due to the definition of $c$ and to the facts that $V(G_{e_1 \oplus e_2})=V(G_{e_1})\cup V(G_{e_2})$ and for every $v\in G_{e_1}$ we have that $\ell_{e_1}(v)=\ell_{e_1 \oplus e_2}(v)$ and for every $v\in G_{e_2}$ we have that $\ell_{e_2}(v)=\ell_{e_1 \oplus e_2}(v)$.
Let us focus on (C2) now. Consider $v\in G_{e_1 \oplus e_2}$. Assume, without loss of generality, that $v\in V(G_{e_1})$.
We will prove that $check(v, c(v), \upcomingNeighbors_1, \ldots, \upcomingNeighbors_\NumberOfColors)$ is true, where $\upcomingNeighbors_j = \min(\mathcal{N}, N[\ell_{e_1 \oplus e_2}(v), a_j] + |\setst{u \in N_{G_{e_1 \oplus e_2}}(v)}{c(u) = a_j}|)$ for every $j \in \intInterval{1}{\NumberOfColors}$.
Since $c_1$ is a $(C_1, N_1)$-coloring of $G_{e_1}$, we know that $check(v, c_1(v), \upcomingNeighbors'_1, \ldots, \upcomingNeighbors'_\NumberOfColors)$ is true, where $\upcomingNeighbors'_j = \min(\mathcal{N}, N_1[\ell_{e_1}(v), a_j] + |\setst{u \in N_{G_{e_1}}(v)}{c_1(u) = a_j}|)$ for every $j \in \intInterval{1}{\NumberOfColors}$.
Using similar arguments as above, we obtain again that $\upcomingNeighbors'_j = \upcomingNeighbors_j$ for every $j \in \intInterval{1}{\NumberOfColors}$.
Due to the definition of $c$, we then conclude that $check(v, c(v), \upcomingNeighbors_1, \ldots, \upcomingNeighbors_\NumberOfColors)$ is true for every $v \in V(G)$, and so (C2) is satisfied.
Thus, $c$ is a $(C, N)$-coloring of $G_{e_1 \oplus e_2}$.
\end{proof}

\noindent\textbf{Lemma \ref{lemma:Join} (Join: $\eta_{i,j}(e)$).}
\emph{Let $N_e \in \intInterval{0}{\mathcal{N}}^{\NumberOfLabels \times \NumberOfColors}$ be such that
\begin{itemize}
\item $N_e[i, a] = \min(\mathcal{N}, N[i, a] + C[j, a])$ for every $a \in \ColorSet$;
\item $N_e[j, a] = \min(\mathcal{N}, N[j, a] + C[i, a])$ for every $a \in \ColorSet$;
\item $N_e[h, a] = N[h, a]$ for every $h \in \intInterval{1}{\NumberOfLabels} \setminus \{i,j\}$ and every $a \in \ColorSet$.
\end{itemize}
Then,
$\lambda(\eta_{i,j}(e), C, N) = \lambda(e, C, N_e)$.
}
\begin{proof}
First of all, since the labelings of the vertices do not change between $G_{\eta_{i,j}(e)}$ and $G_{e}$, we simply use $\ell$ to denote both labelings $\ell_e$ and $\ell_{\eta_{i,j}(e)}$.
Also, since $V(G_{\eta_{i,j}(e)}) = V(G_e)$, we simply denote this set by $V$.

Let $c \colon V \to \ColorSet$ be a valid coloring.
We are going to show that $c$ is a $(C,N)$-coloring of $G_{\eta_{i,j}(e)}$ if and only if $c$ is a  $(C,N_e)$-coloring of $G_{e}$.
In order to prove this, it suffices to show that, for every color $a\in \ColorSet$ and every vertex $v\in V$, we have
$\min(\mathcal{N},N[\ell(v), a] + |\{u \in N_{G_{\eta_{i,j}(e)}}(v) : c(u) = a\}|)
=
\min(\mathcal{N},N_e[\ell(v), a] + |\{u \in N_{G_e}(v) : c(u) = a\}|)$.
If $\ell(v)\neq i,j$ then $N_e[\ell(v), a_i] = N[\ell(v), a_i]$ by definition of $N_e$ and clearly $N_{G_e}(v) = N_{G_{\eta_{i,j}(e)}}(v)$.
On the other hand, if $\ell(v)=i$ (if $\ell(v) = j$ the proof is analogous) then
\begin{align*}
	&	\min(\mathcal{N},N_e[i, a] + |\{u \in N_{G_{e}}(v) : c(u) = a\}|)	\\
=	&	\min(\mathcal{N}, \min(\mathcal{N},N[i,a] + C[j,a]) + |\{u \in N_{G_{e}}(v) : c(u) = a\}|)	\\
=	&	\min(\mathcal{N},N[i,a] + C[j,a] + |\{u \in N_{G_{e}}(v) : c(u) = a\}|)	\\
=	&	\min(\mathcal{N},N[i,a] + \min(\mathcal{N},|\{u \in V : c(u) = a \land \ell(u) = j\}|) + |\{u \in N_{G_{e}}(v) : c(u) = a\}|)	\\
=	&	\min(\mathcal{N},N[i,a] + |\{u \in V : c(u) = a \land \ell(u) = j\}| + |\{u \in N_{G_{e}}(v) : c(u) = a\}|)	\\
=	&	\min(\mathcal{N},N[i,a] + |\{u \in N_{G_{\eta_{i,j}(e)}}(v) : c(u) = a \land \ell(u) = j\}| + |\{u \in N_{G_{e}}(v) : c(u) = a\}|)	\\
=	&	\min(\mathcal{N},N[i, a] + |\{u \in N_{G_{\eta_{i,j}(e)}}(v) : c(u) = a\}|)
\end{align*}

The last two equalities follow from the fact that every vertex with label $j$ is a neighbor of $v$ in $G_{\eta_{i,j}(e)}$ but a non-neighbor of $v$ in $G_{e}$ (recall that we are working with irredundant clique-width expressions).
\end{proof}

\noindent\textbf{Lemma \ref{lemma:Relabeling} (Relabeling: $\rho_{i \rightarrow j}(e)$).}
\emph{Let $N_e$ be such that
\begin{itemize}
\item $N_e[i, a] = N[j, a]$ for every $a \in \ColorSet$;
\item $N_e[h, a] = N[h, a]$ for every $h \in \intInterval{1}{\NumberOfLabels} \setminus \{i\}$ and every $a \in \ColorSet$.
\end{itemize}

If $C[i, a] = 0$ for all $a \in \ColorSet$, then
\begin{align*}
\lambda(\rho_{i \rightarrow j}(e), C, N) =
	\min\{ &
		\lambda(e, C_e, N_e)
		:\\
		& \text{(a)}\; C_e \in \intInterval{0}{\mathcal{N}}^{\NumberOfLabels \times \NumberOfColors}\\
		& \text{(b)}\; C[j,a] = \min(\mathcal{N}, C_e[i, a] + C_e[j, a]) \text{ for all } a \in \ColorSet;\\
		& \text{(c)}\; C_e[h, a] = C[h, a] \text{ for all } h \in \intInterval{1}{\NumberOfLabels} \setminus \{i,j\}, a \in \ColorSet
	\}.
\end{align*}
Otherwise, $\lambda(\rho_{i \rightarrow j}(e), C, N) = \Error$.
}
\begin{proof}
If $C[i, a] \neq 0$ for some $a \in \ColorSet$, then, by definition, there exists no $(C, N)$-coloring of $G_{\rho_{i \rightarrow j}(e)}$ and $\lambda(\rho_{i \rightarrow j}(e), C, N) = \Error$.
So we may assume now that $C[i, a] = 0$ for all $a \in \ColorSet$.
Let $\alpha=\min\{\lambda(e, C_e, N_e) : (a),(b), (c)$ are satisfied$\}$.
We will first prove that $\lambda(\rho_{i \rightarrow j}(e), C, N)\geq \alpha$.
Let $c$ be a $(C,N)$-coloring of $G_{\rho_{i \rightarrow j}(e)}$ whose weight equals $\lambda(\rho_{i \rightarrow j}(e), C, N)$.
We will show that there exists a matrix $C_e$ in $\intInterval{0}{\mathcal{N}}^{\NumberOfLabels \times \NumberOfColors}$ satisfying $(b),(c)$ and such that $c$ is a $(C_e,N_e)$-coloring of $G_{e}$.
Let 
$$C_e[h,a] = \min(\mathcal{N}, |\{v \in V({G_e}) : c(v) = a \land \ell_e(v) = h\}|)$$
for every $h\in \intInterval{1}{\NumberOfLabels}$ and every $a\in \ColorSet$.
Since $c$ is a $(C,N)$-coloring of $G_{\rho_{i \rightarrow j}(e)}$, then $C[h,a] = \min(\mathcal{N}, |\{v \in V(G_{\rho_{i \rightarrow j}(e)}) : c(v) = a \land \ell_{\rho_{i \rightarrow j}(e)}(v) = h\})$ for every $h\in \intInterval{1}{\NumberOfLabels}$ and $a\in \ColorSet$.
Therefore, $(c)$ is trivially satisfied, since $\ell_{\rho_{i \rightarrow j}(e)}(v)=\ell_e(v)$ for any vertex $v$ whose label is neither $i$ nor $j$ in $G_e$.
Furthermore the following inequalities hold,
\begin{align*}
\min(\mathcal{N}, C_e[i,a] + C_e[j,a]) =&\min(\mathcal{N}, \min(\mathcal{N}, |\{v \in V({G_e}) : c(v) = a \land \ell_e(v) = i\}|)\\
&+\min(\mathcal{N}, |\{v \in V({G_e}) : c(v) = a \land \ell_e(v) = j\}|))\\
=&\min(\mathcal{N}, |\{v \in V({G_e}) : c(v) = a \land \ell_e(v) = i\}|\\
& +|\{v \in V({G_e}) : c(v) = a \land \ell_e(v) = j\}|)\\
=&\min(\mathcal{N}, |\{v \in V({G_e}) : c(v) = a \land (\ell_e(v) = i \lor\ell_e(v) = j)\}|)\\
=&\min(\mathcal{N}, |\{v \in V({G_{\rho_{i \rightarrow j}(e)}}) : c(v) = a \land \ell_{\rho_{i \rightarrow j}(e)}(v) = j\}|)\\
=& C[j,a]
\end{align*}
and thus, condition $(b)$ is satisfied as well.
We next show that $c$ is a $(C_e,N_e)$-coloring of $G_{e}$.
We know that $c$ is a valid coloring because $c$ is a $(C,N)$-coloring of $G_{\rho_{i \rightarrow j}(e)}$.
Also, (C1) is trivially satisfied from our definition of $C_e$.
For (C2), we have to show that $check(v, c(v), \upcomingNeighbors'_1, \ldots, \upcomingNeighbors'_{\NumberOfColors})$ is true for every $v$ in $G_e$, where $\upcomingNeighbors'_b=\min(\mathcal{N},N_e[\ell_e(v),a_b]+|\{u \in N_{G_e}(v) : c(u) = a_b\}|)$ for every $v\in V(G_e)$ and every $b\in \intInterval{1}{\NumberOfColors}$.
Since $c$ is a $(C,N)$-coloring of $G_{\rho_{i \rightarrow j}(e)}$, we know that $check(v, c(v), \upcomingNeighbors_1, \ldots, \upcomingNeighbors_\NumberOfColors)$ is true for every $v\in V(G_{\rho_{i \rightarrow j}(e)})$, where $\upcomingNeighbors_b=\min(\mathcal{N}, N[\ell_{\rho_{i \rightarrow j}(e)}(v), a_b] + |\{u \in N_{G_{\rho_{i \rightarrow j}(e)}}(v) : c(u) = a_b\}|)$ for every $b \in \intInterval{1}{\NumberOfColors}$.
By definition of $N_e$ and since $N_{G_e}(v) = N_{G_{\rho_{i \rightarrow j}(e)}}(v)$ for all vertices $v$, we have $\upcomingNeighbors_b=\upcomingNeighbors'_b$ for all $b\in \intInterval{1}{\NumberOfColors}$, and therefore $check(v, c(v), \upcomingNeighbors'_1, \ldots, \upcomingNeighbors'_\NumberOfColors)$ is true for every $v\in G_e$ and $b\in \intInterval{1}{\NumberOfColors}$.

We now show that $\lambda(\rho_{i \rightarrow j}(e), C, N)\leq \alpha$.
Let $C_e$ be a matrix in $ \intInterval{0}{\mathcal{N}}^{\NumberOfLabels \times \NumberOfColors}$ satisfying $(b),(c)$, and let $c$ be a $(C_e, N_e)$-coloring of $G_{e}$ of weight $\lambda(e, C_e, N_e)$.
We are going to show that $c$ is also a $(C,N)$-coloring of $G_{\rho_{i \rightarrow j}(e)}$.
Condition (C1) is trivially satisfied for every label $h$ in $\intInterval{1}{\NumberOfLabels}\setminus \{i, j\}$ and every color $a$.
For label $i$, condition (C1) is also true since $C[i,a]=0$ by assumption and there are no vertices with label $i$ in $V(G_{\rho_{i \rightarrow j}(e)})$.
We can show in a similar way as above that $C[j,a]=\min(\mathcal{N}, |\{v \in V(G_{\rho_{i \rightarrow j}(e)}) : c(v) = a \land \ell_{\rho_{i \rightarrow j}(e)}(v) = j\}|)$ for all $a \in \ColorSet$.
We need to verify now (C2), i.e., for every vertex $v$, we have that $check(v,c(v),\upcomingNeighbors_1,\ldots,\upcomingNeighbors_\NumberOfColors)$ is true, where $\upcomingNeighbors_b=\min(\mathcal{N},N[\ell_{\rho_{i \rightarrow j}(e)}(v),a_b]+|\{u \in N_{G_{\rho_{i \rightarrow j}(e)}}(v) : c(u) = a_b\}|)$.
As before, this is a consequence of the fact that $c$ is a $(C_e,N_e)$-coloring of $G_{e}$, and that for every vertex $v$ and color $a_b$, we have
$N_e[\ell_e(v),a_b]=N[\ell_{\rho_{i \rightarrow j}(e)}(v),a_b]$ by definition and $|\{u \in N_{G_e}(v) : c(u) = a_b\}|)=|\{u \in N_{G_{\rho_{i \rightarrow j}(e)}}(v) : c(u) = a_b\}|$.
\end{proof}

\subsection{Omitted proofs of Section~\ref{sec:SizeGlobProp}}
In what follows, we will write $\widehat{p}$ instead of $p_1, \ldots, p_m$ to make the notation less cumbersome. Note that $\widehat{p}$ is empty when $m = 0$.

\begin{lemma}\label{lem:GP-size-initial}
Let $\Pi$ be color-counting 1-locally checkable problem with a set of global properties $\Gamma$, input graph $G$ with a clique-width $\NumberOfLabels$-expression $e_G$ of $G$, and let $a \in \ColorSet$.
Let $\sigma \subseteq \mathbb{N}$ and let $(Q, \{1\}, \delta, q_0, F)$ be a deterministic finite-state automaton that accepts a string of $n$ consecutive 1's if and only if $n \in \sigma$.
Let $\in_F \colon Q \to \BoolSet$ be the function such that $\in_F(q) = (q \in F)$.

Assume that the minimum weight of a proper coloring of $G$ satisfying $\Gamma$ equals
$$\min\setst{\lambda(e_G, C, N, \widehat{p})}{P(C, N, \widehat{p})=\True}$$
for some property $P$.
Furthermore, assume that
\begin{itemize}
\item[(1)] for every proper coloring $c$ of $G$ satisfying $\Gamma$ there exist $C, N, \widehat{p}$ such that $P(C, N, \widehat{p})=\True$ and such that $c$ is a $(C, N, \widehat{p})$-coloring of $G$;

\item[(2)] for every $C, N, \widehat{p}$ such that $P(C, N, \widehat{p})=\True$ and such that $\lambda(e_G, C, N, \widehat{p}) \neq \Error$, every $(C, N, \widehat{p})$-coloring of $G$ is a proper coloring of $G$ satisfying $\Gamma$.
\end{itemize}

Then, the minimum weight of a proper coloring $c$ of $G$ satisfying $\Gamma$ and such that $|\setst{v \in V(G)}{c(v) = a}| \in \sigma$ equals
$$\min\setst{\lambda(e_G, C, N, \widehat{p}, q_0, \in_F)}{P(C, N, \widehat{p})=\True}.$$
Moreover,
\begin{itemize}
\item[(a)] for every proper coloring $c$ of $G$ satisfying $\Gamma$ and such that $|\setst{v \in V(G)}{c(v) = a}| \in \sigma$, there exist $C, N, \widehat{p}$ such that $P(C, N, \widehat{p})=\True$ and such that $c$ is $(C, N, \widehat{p}, q_0, \in_F)$-coloring of $G$,

\item[(b)] for every $C, N, \widehat{p}$ such that $P(C, N, \widehat{p})=\True$ and such that $\lambda(e_G, C, N, \widehat{p}, q_0, \in_F) \neq \Error$, every $(C, N, \widehat{p}, q_0, \in_F)$-coloring $c$ of $G$ is a proper coloring of $G$ satisfying $\Gamma$ and such that $|\setst{v \in V(G)}{c(v) = a}| \in \sigma$.
\end{itemize}

\end{lemma}
\begin{proof}
We will first show that for every proper coloring $c$ of $G$ satisfying $\Gamma$ and such that $|\setst{v \in V(G)}{c(v) = a}| \in \sigma$, there exist $C, N, \widehat{p}$ such that $P(C, N, \widehat{p})=\True$ and such that $\textsc{w}(c) \geq \lambda(e_G, C, N, \widehat{p}, q_0, \in_F)$.
Suppose we have such a proper coloring $c$ of $G$ satisfying $\Gamma$ and such that $|\setst{v \in V(G)}{c(v) = a}| \in \sigma$.
By assumption (1), we know that there exist $C, N, \widehat{p}$ such that $P(C, N, \widehat{p})=\True$ and $c$ is a $(C, N, \widehat{p})$-coloring of $G$.
Since we are assuming that $|\setst{v \in V(G)}{c(v) = a}| \in \sigma$, and since the automaton accepts a string of $t$ consecutive 1's if and only if $t \in \sigma$, it follows that $\in_F(\delta^n(q_0)) = \True$, where $n = |\setst{v \in V(G)}{c(v) = a}|$.
Thus, $c$ is a $(C, N, \widehat{p}, q_0, \in_F)$-coloring of $G$ and $\textsc{w}(c) \geq \lambda(C, N, \widehat{p}, q_0, \in_F)$.

On the other hand, we will now show that for every $C, N, \widehat{p}$ such that $P(C, N, \widehat{p}) = \True$ and such that $\lambda(e_G, C, N, \widehat{p}, q_0, \in_F) \neq \Error$, there exists a proper coloring $c$ of $G$ satisfying $\Gamma$ and such that $|\setst{v \in V(G)}{c(v) = a}| \in \sigma$ with $\textsc{w}(c) = \lambda(e_G, C, N, \widehat{p}, q_0, \in_F)$.
So suppose we have $C, N, \widehat{p}$ such that $P(C, N, \widehat{p}) = \True$ and such that $\lambda(C, N, \widehat{p}, q_0, \in_F) \neq \Error$.
By the latter assumption and by the definition of $\lambda(C, N, \widehat{p}, q_0, \in_F)$, we get that $\lambda(C, N, \widehat{p}) \neq \Error$.
Let $c$ be a $(C, N, \widehat{p}, q_0, \in_F)$-coloring of $G$ (notice that at least one such $c$ exists).
By definition, $c$ is a $(C, N, \widehat{p})$-coloring of $G$.
Then we conclude by assumption (2) that $c$ is a proper coloring of $G$ satisfying $\Gamma$.
Finally, since $c$ is a $(C, N, \widehat{p}, q_0, \in_F)$-coloring of $G$, we have that $\in_F(\delta^n(q_0)) = \True$, where $n = |\setst{v \in V(G)}{c(v) = a}|$, which implies that $|\setst{v \in V(G)}{c(v) = a}| \in \sigma$.
If we consider in particular $c$ of minimum weight, then $\textsc{w}(c) = \lambda(e_G, C, N, \widehat{p}, q_0, \in_F)$.

Notice that (a) and (b) are implicitly shown by the above.
\end{proof}

\begin{lemma}[Creating new vertex: $i(v)$]\label{lem:GP-size-newVertex}
If $C[i,a]=1$ and $f_a(\delta(s_a, 1))=\True$,
or if $C[i,a]=0$ and $f_a(s_a)=\True$, then
	\begin{align*}
	\genDomSym(i(v), C, N,\widehat{p}, s_a, f_a) = \genDomSym(i(v), C, N,\widehat{p}) .
	\end{align*}
Otherwise, $\genDomSym(i(v), C, N,\widehat{p}, s_a, f_a) = \Error$.

\end{lemma}
\begin{proof}
Clearly,  $G_{i(v)}$ is a graph consisting of a single vertex $v$ with label $i$.
By definition, $\genDomSym(i(v), C, N,$ $\widehat{p}, s_a, f_a)$ is the minimum weight among all  $(C, N, \widehat{p})$-colorings $c$ of $G_{i(v)}$ such that $f_a(\delta^n (s_a)) = \True$, where $n = |\{u \in V(G_{i(v)}) : c(u) = a\}|$.
Moreover, any $(C, N,\widehat{p})$-coloring $c$ of $G_{i(v)}$ satisfies the following property: $\min(\mathcal{N}, |\setst{u \in V(G_{i(v)})}{c(u) = b \land \ell_{i(v)}(u) = j}|) = C[j,b]$ for all $j \in \intInterval{1}{\NumberOfLabels}$ and all $b \in \ColorSet$ (recall Definition~\ref{def:CN-coloring} of a $(C,N)$-coloring).
In particular, since $\mathcal{N} \geq 1$ and the graph $G_{i(v)}$ has only one vertex, any $(C, N,\widehat{p})$-coloring $c$ of $G_{i(v)}$ satisfies the following property: if $C[i,a] = 1$ then $c(v) = a$, otherwise $c(v) \neq a$.
Therefore, we can restate the definition of $\genDomSym(i(v), C, N,\widehat{p}, s_a, f_a)$ as the minimum weight among all $(C, N,\widehat{p})$-colorings $c$ of $G_{i(v)}$ such that if $C[i,a] = 1$ then $f_a(\delta(s_a)) = \True$, otherwise $f_a(s_a) = \True$.
Now, since $\genDomSym(i(v), C, N,\widehat{p})$ is the minimum weight among all $(C, N,\widehat{p})$-colorings of $G_{i(v)}$, the statement trivially follows.
\end{proof}

\begin{lemma}[Disjoint union: $e_1 \oplus e_2$]\label{lem:GP-size-disjointUnion}
Assume that
\begin{align*}
\genDomSym(e_1 \oplus e_2, C, N, \widehat{p}) =
	\min\{&
		\genDomSym(e_1, C_1, N_1, \widehat{p_1})
		\weightsSum
		\genDomSym(e_2, C_2, N_2, \widehat{p_2})
	:\\
	&P(C, N, \widehat{p}, C_1, N_1, \widehat{p_1}, C_2, N_2, \widehat{p_2}) = \True
\}
\end{align*}
for some property $P$.
Moreover, assume that 
\begin{itemize}
\item[(1)] for every $(C, N, \widehat{p})$-coloring $c$ of $G_{e_1 \oplus e_2}$ there exist $C_1, N_1, \widehat{p_1}, C_2, N_2, \widehat{p_2}$ such that $P(C, N,$ $\widehat{p}, C_1, N_1, \widehat{p_1}, C_2, N_2, \widehat{p_2}) = \True$ and such that $c|_{V(G_{e_1})}$ is a $(C_1, N_1, \widehat{p_1})$-coloring of $G_{e_1}$ and $c|_{V(G_{e_2})}$ is a $(C_2, N_2, \widehat{p_2})$-coloring of $G_{e_2}$;

\item[(2)] for all $C_1, N_1, \widehat{p_1}, C_2, N_2, \widehat{p_2}$ such that $P(C, N, \widehat{p}, C_1, N_1, \widehat{p_1}, C_2, N_2, \widehat{p_2}) = \True$, if $c_1$ is a $(C_1, N_1, \widehat{p_1})$-coloring of $G_{e_1}$ and $c_2$ is a $(C_2, N_2, \widehat{p_2})$-coloring of $G_{e_2}$, then $c = c_1 \cup c_2$ is a $(C, N, \widehat{p})$-coloring of $G_{e_1 \oplus e_2}$.
\end{itemize}

Then,
\begin{align*}
\genDomSym(e_1 \oplus e_2, C, N, \widehat{p}, s_a, f_a) =
	\min\{&
		\genDomSym(e_1, C_1, N_1, \widehat{p_1}, s_a, eq_{q})
		\weightsSum
		\genDomSym(e_2, C_2, N_2, \widehat{p_2}, q, f_a)
	:\\
	&q \in Q
	\text{ and }
	P(C, N, \widehat{p}, C_1, N_1, \widehat{p_1}, C_2, N_2, \widehat{p_2}) = \True
\}.
\end{align*}

Moreover,
\begin{itemize}
\item[(a)] for every $(C, N, \widehat{p}, s_a, f_a)$-coloring $c$ of $G_{e_1 \oplus e_2}$ there exist $q$, $C_1, N_1, \widehat{p_1}, C_2, N_2, \widehat{p_2}$ such that $q \in Q$, $P(C, N, \widehat{p}, C_1, N_1, \widehat{p_1}, C_2, N_2, \widehat{p_2}) = \True$, and $c|_{V(G_{e_1})}$ is a $(C_1, N_1, \widehat{p_1}, s_a, eq_q)$-coloring of $G_{e_1}$ and $c|_{V(G_{e_2})}$ is a $(C_2, N_2, \widehat{p_2}, q, f_a)$-coloring of $G_{e_2}$;

\item[(b)] for all $q$, $C_1, N_1, \widehat{p_1}, C_2, N_2, \widehat{p_2}$ such that $q \in Q$ and $P(C, N, \widehat{p}, C_1, N_1, \widehat{p_1}, C_2, N_2, \widehat{p_2}) = \True$, if $c_1$ is a $(C_1, N_1, \widehat{p_1}, s_a, eq_q)$-coloring of $G_{e_1}$ and $c_2$ is a $(C_2, N_2, \widehat{p_2}, q, f_a)$-coloring of $G_{e_2}$ then $c = c_1 \cup c_2$ is a $(C, N, \widehat{p}, s_a, f_a)$-coloring of $G_{e_1 \oplus e_2}$.
\end{itemize}

\end{lemma}
\begin{proof}
Let
$\alpha = \min\{
	\genDomSym(e_1, C_1, N_1, \widehat{p_1}, s_a, eq_{q})
	\weightsSum
	\genDomSym(e_2, C_2, N_2, \widehat{p_2}, q, f_a)
	:
	q \in Q
	\text{ and }
	P(C, N, \widehat{p}, C_1, N_1,$ $\widehat{p_1}, C_2, N_2, \widehat{p_2}) = \True
\}$.
We will first prove that $\genDomSym(e_1 \oplus e_2, C, N, \widehat{p}, s_a, f_a) \geq \alpha$.
Note that if $\lambda(e_1 \oplus e_2, C, N, \widehat{p}, s_a, f_a) = \Error$, then we are done.
So assume that $\lambda(e_1 \oplus e_2, C, N, \widehat{p}, s_a, f_a) \neq \Error$.
Let $c$ be a $(C, N, \widehat{p}, s_a, f_a)$-coloring of $G_{e_1 \oplus e_2}$, that is, $c$ is a $(C, N, \widehat{p})$-coloring of $G_{e_1 \oplus e_2}$ such that  $f_a(\delta^n (s_a)) = \True$, where $n = |\{v \in V(G_{e_1 \oplus e_2}) : c(v) = a\}|$.
Let $c_1 = c|_{V(G_{e_1})}$ and $c_2 = c|_{V(G_{e_2})}$.
We need to show that there exist $C_1, N_1, \widehat{p_1}, C_2, N_2, \widehat{p_2}$ and $q \in Q$ such that $P(C, N, \widehat{p}, C_1, N_1, \widehat{p_1}, C_2, N_2, \widehat{p_2})=\True$, and $c_1$ is a $(C_1, N_1, \widehat{p_1}, s_a, eq_q)$-coloring of $G_{e_1}$, and $c_2$ is a $(C_2, N_2, \widehat{p_2}, q, f_a)$-coloring of $G_{e_2}$.
By assumption (1), we know already that there exist $C_1, N_1, \widehat{p_1}, C_2, N_2, \widehat{p_2}$ such that $P(C, N, \widehat{p}, C_1, N_1, \widehat{p_1}, C_2, N_2, \widehat{p_2})=\True$, and $c_1$ is a $(C_1, N_1, \widehat{p_1})$-coloring of $G_{e_1}$ and $c_2$ is a $(C_2, N_2, \widehat{p_2})$-coloring of $G_{e_2}$.
Let $n_1 = |\{v \in V(G_{e_1}) : c_1(v) = a\}|$ and $n_2 = |\{v \in V(G_{e_2}) : c_2(v) = a\}|$.
Clearly, $n_1 + n_2 = n$, since $V(G_{e_1 \oplus e_2})=V(G_{e_1})\cup V(G_{e_2})$ and $V(G_{e_1})$ and $V(G_{e_2})$ are disjoint.
Let $q = \delta^{n_1}(s_a)$.
Clearly, $eq_q(\delta^{n_1}(s_a)) = \True$, and thus, $c_1$  is a $(C_1, N_1, \widehat{p_1},  s_a, eq_{q})$-coloring of $G_{e_1}$.
Furthermore, $f_a(\delta^{n_2}(q)) = f_a(\delta^{n_2}(\delta^{n_1}(s_a))) = f_a(\delta^{n_1 + n_2}(s_a)) = f_a(\delta^n(s_a)) = \True$, where the last equality comes from the fact that $c$ is a $(C, N, \widehat{p}, s_a, f_a)$-coloring of $G_{e_1 \oplus e_2}$.
We conclude that $c_2$ is a $(C_2, N_2, \widehat{p_2}, q, f_a)$-coloring of $G_{e_2}$.
If we consider in particular a $(C, N, \widehat{p}, s_a, f_a)$-coloring $c$ of $G_{e_1 \oplus e_2}$ of minimum weight, then $\genDomSym(e_1 \oplus e_2, C, N, \widehat{p}, s_a, f_a) = \textsc{w}(c) = \textsc{w}(c_1) \weightsSum \textsc{w}(c_2) \geq \genDomSym(e_1, C_1, N_1, \widehat{p_1}, s_a, eq_{q}) \weightsSum \genDomSym(e_2, C_2, N_2, \widehat{p_2}, q, f_a) \geq \alpha$.

Now we will show that $\genDomSym(e_1 \oplus e_2, C, N, \widehat{p}, s_a, f_a) \leq \genDomSym(e_1, C_1, N_1, \widehat{p_1}, s_a, eq_{q}) \weightsSum \genDomSym(e_2, C_2, N_2,$ $\widehat{p_2}, q, f_a)$ whenever $q \in Q$ and $P(C, N, \widehat{p}, C_1, N_1, \widehat{p_1}, C_2, N_2, \widehat{p_2}) = \True$.
This will directly imply that $\genDomSym(e_1 \oplus e_2, C, N, \widehat{p}, s_a, f_a) \leq \alpha$.
Let $q \in Q$ and $C_1, N_1, \widehat{p_1}, C_2, N_2, \widehat{p_2}$ be such that $P(C, N, \widehat{p}, C_1, N_1, \widehat{p_1}, C_2,$ $N_2, \widehat{p_2}) = \True$.
If $\genDomSym(e_1, C_1, N_1, \widehat{p_1}, s_a, eq_{q}) = \Error$ or $\genDomSym(e_2, C_2, N_2, \widehat{p_2}, q, f_a) = \Error$, then we are done.
So assume that $\genDomSym(e_1, C_1, N_1, \widehat{p_1}, s_a, eq_{q}) \neq \Error$ and $\genDomSym(e_2, C_2, N_2, \widehat{p_2}, q, f_a) \neq \Error$.
Let $c_1$ be a $(C_1, N_1, \widehat{p_1}, s_a, eq_q)$-coloring of $G_{e_1}$ and 
$c_2$ be a $(C_2, N_2, \widehat{p_2}, q, f_a)$-coloring of $G_{e_2}$.
Let $n_1 = |\{v \in V(G_{e_1}) : c_1(v) = a\}|$ and $n_2 = |\{v \in V(G_{e_2}) : c_2(v) = a\}|$.
Consider $c = c_1 \cup c_2$ which, by assumption (2), is a $(C, N, \widehat{p})$-coloring of $G_{e_1 \oplus e_2}$.
We clearly have that $\textsc{w}(c) = \textsc{w}(c_1) \weightsSum \textsc{w}(c_2)$.
To show that $c$ is a $(C, N, \widehat{p}, s_a, f_a)$-coloring of $G_{e_1 \oplus e_2}$, it remains to show that $f_a(\delta^n(s_a))=\True$, where $n = |\{v \in V(G_{e_1 \oplus e_2}) : c(v) = a\}|$.
Clearly, $n = n_1 + n_2$.
Since $c_1$ is a $(C_1, N_1, \widehat{p_1}, s_a, eq_q)$-coloring of $G_{e_1}$, we have $eq_q(\delta^{n_1}(s_a)) = \True$, thus, $q = \delta^{n_1}(s_a)$.
Similarly, since $c_2$ is a $(C_2, N_2, \widehat{p_2}, q, f_a)$-coloring of $G_{e_2}$, we have that $f_a(\delta^{n_2}(q)) = \True$.
Hence, $\True = f_a(\delta^{n_2}(q)) = f_a(\delta^{n_2}(\delta^{n_1}(s_a))) = f_a(\delta^{n_1 + n_2}(s_a)) = f_a(\delta^{n}(s_a))$, and thus, $c$ is a $(C, N, \widehat{p}, s_a, f_a)$-coloring of $G_{e_1 \oplus e_2}$.
If we consider in particular a $(C_1, N_1, \widehat{p_1}, s_a, eq_q)$-coloring $c_1$ of $G_{e_1}$ of minimum weight as well as a $(C_2, N_2, \widehat{p_2}, q, f_a)$-coloring $c_2$ of $G_{e_2}$ of minimum weight, we obtain $\genDomSym(e_1, C_1, N_1, \widehat{p_1}, s_a, eq_{q}) \weightsSum \genDomSym(e_2, C_2, N_2, \widehat{p_2}, q, f_a) = \textsc{w}(c_1) \weightsSum \textsc{w}(c_2) = \textsc{w}(c) \geq \genDomSym(e_1 \oplus e_2, C, N, \widehat{p}, s_a, f_a)$.

Notice that (a) and (b) are shown implicitly by the above.
\end{proof}

\begin{lemma}[Join: $\eta_{i,j}(e)$]\label{lem:GP-size-join}
Assume that there exist $C', N', \widehat{p'}$ such that $c$ is a $(C, N, \widehat{p})$-coloring of $G_{\eta_{i,j}(e)}$ if and only if $c$ is a $(C', N', \widehat{p'})$-coloring of $G_{e}$.

Then, $c$ is a $(C, N, \widehat{p}, s_a, f_a)$-coloring of $G_{\eta_{i,j}(e)}$ if and only if $c$ is a $(C', N', \widehat{p'}, s_a, f_a)$-coloring of $G_{e}$.
In particular,
$$\genDomSym(\eta_{i,j}(e), C, N, \widehat{p}, s_a, f_a) = \genDomSym(e, C', N', \widehat{p'}, s_a, f_a).$$

\end{lemma}
\begin{proof}
Let $c$ be a $(C, N, \widehat{p}, s_a, f_a)$-coloring of $G_{\eta_{i,j}(e)}$, i.e. $c$ is a $(C, N, \widehat{p})$-coloring of $G_{\eta_{i,j}(e)}$ such that $f_a(\delta^n(s_a))=\True$, where $n = |\{v \in V(G_{\eta_{i,j}(e)}) : c(v) = a\}|$.
By assumption, we know there exist $C', N', \widehat{p'}$ such that $c$ is also a $(C', N', \widehat{p'})$-coloring of $G_{e}$.
It remains to show that $f_a(\delta^{n'}(s_a))=\True$, where $n' = |\{v \in V(G_{e}) : c(v) = a\}|$.
But this trivially follows from the fact $n = n'$.
Thus, $c$ is also a $(C', N', \widehat{p'}, s_a, f_a)$-coloring of $G_{e}$.
The reverse can be shown similarly.
\end{proof}

\begin{lemma}[Relabeling: $\rho_{i \rightarrow j}(e)$]\label{lem:GP-size-relabeling}
Assume that
$$\genDomSym(\rho_{i \rightarrow j}(e), C, N, \widehat{p}) =
\min\{
	\genDomSym(e, C_e, N_e, \widehat{p_e})
	:
	P(C, N, \widehat{p}, C_e, N_e, \widehat{p_e}) =\True
\}$$
for some property $P$.
Moreover, assume that
\begin{itemize}
\item[(1)] for every $(C, N, \widehat{p})$-coloring $c$ of $G_{\rho_{i \rightarrow j}(e)}$, there exist parameters $C_e, N_e, \widehat{p_e}$ such that $P(C, N, \widehat{p},$ $C_e, N_e, \widehat{p_e})=\True$ and $c$ is a $(C_e, N_e, \widehat{p_e})$-coloring of $G_e$;

\item[(2)] for all parameters $C_e, N_e, \widehat{p_e}$ such that $P(C, N, \widehat{p}, C_e, N_e, \widehat{p_e})=\True$, if $c$ is a $(C_e, N_e, \widehat{p_e})$-coloring of $G_e$, then $c$ is also a $(C, N, \widehat{p})$-coloring $c$ of $G_{\rho_{i \rightarrow j}(e)}$.
\end{itemize}

Then,
$$\genDomSym(\rho_{i \rightarrow j}(e), C, N, \widehat{p}, s_a, f_a) =
\min\{
	\genDomSym(e, C_e, N_e, \widehat{p_e}, s_a, f_a)
	:
	P(C, N, \widehat{p}, C_e, N_e, \widehat{p_e})=\True
\}.$$
Moreover,
\begin{itemize}
\item[(a)] for every $(C, N,\widehat{p}, s_a, f_a)$-coloring $c$ of $G_{\rho_{i \rightarrow j}(e)}$, there exist parameters $C_e, N_e, \widehat{p_e}$ such that $P(C, N, \widehat{p}, C_e, N_e,\widehat{p_e})=\True$ and $c$ is a $(C_e, N_e,\widehat{p_e}, s_a, f_a)$-coloring of $G_e$;

\item[(b)] for all parameters $C_e, N_e,\widehat{p_e}$ such that $P(C, N, \widehat{p}, C_e, N_e,\widehat{p_e})=\True$, if $c$ is a $(C_e, N_e,\widehat{p_e},$ $s_a, f_a)$-coloring of $G_e$, then $c$ is also a $(C, N, \widehat{p}, s_a, f_a)$-coloring $c$ of $G_{\rho_{i \rightarrow j}(e)}$.
\end{itemize}

\end{lemma}
\begin{proof}
Let $\alpha = \min\{ \genDomSym(e, C_e, N_e, \widehat{p_e}, s_a, f_a) : P(C, N, \widehat{p}, C_e, N_e, \widehat{p_e}) = \True \}$.
We will first prove that $\genDomSym(\rho_{i \rightarrow j}(e), C, N, \widehat{p}, s_a, f_a)\geq \alpha$.
Let $c$ be a $(C, N,\widehat{p}, s_a, f_a)$-coloring of $G_{\rho_{i \rightarrow j}(e)}$.
We show that there exist parameters $C_e, N_e,\widehat{p_e}$ such that $P(C, N, \widehat{p}, C_e, N_e,\widehat{p_e})=\True$ and $c$ is a $(C_e, N_e,\widehat{p_e}, s_a, f_a)$-coloring of $G_e$.
By definition, $c$ is a $(C, N, \widehat{p})$-coloring of $G_{\rho_{i \rightarrow j}(e)}$ such that $f_a(\delta^{n}(s_a)) = \True$, where $n = |\{v \in V(G_{\rho_{i \rightarrow j}(e)}) : c(v) = a\}|$.
Therefore, by assumption (1), there exist parameters $C_e, N_e, \widehat{p_e}$ such that $P(C, N, \widehat{p}, C_e, N_e, \widehat{p_e})=\True$ and $c$ is a $(C_e, N_e,\widehat{p_e})$-coloring of $G_e$.
Furthermore, since $n_e = |\{v \in V(G_{e}) : c(v) = a\}| = |\{v \in V(G_{\rho_{i \rightarrow j}(e)}) : c(v) = a\}| = n$, it immediately follows that $f_a(\delta^{n_e}(s_a)) = \True$.
Thus, $c$ is a $(C_e, N_e,\widehat{p_e}, s_a, f_a)$-coloring of $G_e$.
If we consider in particular a $(C, N,\widehat{p}, s_a, f_a)$-coloring $c$ of $G_{\rho_{i \rightarrow j}(e)}$ of minimum weight, then $\genDomSym(\rho_{i \rightarrow j}(e), C, N, \widehat{p}, s_a, f_a) = \textsc{w}(c) \geq \genDomSym(e, C_e, N_e,\widehat{p_e}, s_a, f_a) \geq \alpha$.

Let us now show that $\genDomSym(\rho_{i \rightarrow j}(e), C, N, \widehat{p}, s_a, f_a) \leq \alpha$.
Let $C_e, N_e, \widehat{p_e}$ be such that $P(C, N, \widehat{p}, C_e,$ $N_e, \widehat{p_e}) = \True$.
Let $c$ be a $(C_e, N_e,\widehat{p_e}, s_a, f_a)$-coloring of $G_e$.
By definition, $c$ is a $(C_e, N_e, \widehat{p_e})$-coloring of $G_e$ such that $f_a(\delta^{n_e}(s_a)) = \True$, where $n_e = |\{v \in V(G_{e}) : c(v) = a\}|$.
Then, by assumption (2), $c$ is also a $(C, N, \widehat{p})$-coloring $c$ of $G_{\rho_{i \rightarrow j}(e)}$.
Furthermore, since $n = |\{v \in V(G_{\rho_{i \rightarrow j}(e)}) : c(v) = a\}| = |\{v \in V(G_{e}) : c(v) = a\}| = n_e$, it immediately follows that $f_a(\delta^{n}(s_a)) = \True$.
Therefore, $c$ is a $(C, N,\widehat{p}, s_a, f_a)$-coloring $c$ of $G_{\rho_{i \rightarrow j}(e)}$.
If we consider in particular a $(C_e, N_e,\widehat{p_e}, s_a, f_a)$-coloring $c$ of $G_e$ for which $\alpha$ is obtained, then $\alpha = \textsc{w}(c) \geq \genDomSym(\rho_{i \rightarrow j}(e), C, N, \widehat{p}, s_a, f_a)$.

Notice that (a) and (b) are implicitly shown by the above.
\end{proof}

\subsubsection{Complexity of the modified algorithm}
First, assume that for any possible parameter $f_a$ and state $q$, $\delta(s, 1)$ and $f_a(q)$ can be computed in constant time.
Also, assume that we want to fix the size of $\mathcal{R}$ color classes $a_1, \ldots, a_{\mathcal{R}}$.
Let $\mathcal{S}$ be the size of the largest set of states among the $\mathcal{R}$ considered automata.
Then, we add a term $\mathcal{R}$ to the complexity corresponding to the operation of creating a new labeled vertex, and we multiply by a term $\mathcal{S}^{\mathcal{R}}$ the complexity corresponding to the disjoint union operation.
Moreover, whenever we go through all the possible $\genDomSym(e, C, N, \widehat{p}, s_{a_1}, f_{a_1}, \ldots, s_{a_\mathcal{R}}, f_{a_\mathcal{R}})$, we multiply the complexity by a factor $(\mathcal{S} (\mathcal{S}+1))^{\mathcal{R}}$.
Hence, since $\mathcal{R}$ is at most the number of colors and $\mathcal{S} \leq |V(G)|$, we conclude that the new algorithm is also XP parameterized by clique-width when the number of colors is constant.

\end{document}